\documentclass{article}
\usepackage{iclr2027_conference,times}
\usepackage{hyperref}
\usepackage{url}
\usepackage{fontawesome5}

\iclrfinalcopy
\AddToHook{cmd/maketitle/after}{\lhead{Preprint}}

\usepackage{amsmath,amsfonts,bm}

\newcommand{\captiona}{{\em (a)}}
\newcommand{\captionb}{{\em (b)}}
\newcommand{\captionc}{{\em (c)}}

\def\Figref#1{Figure~\ref{#1}}

\def\Secref#1{Section~\ref{#1}}
\def\eqref#1{equation~\ref{#1}}
\def\Eqref#1{Equation~\ref{#1}}
\def\Algref#1{Algorithm~\ref{#1}}

\def\1{\bm{1}}

\def\vone{{\bm{1}}}

\def\vb{{\bm{b}}}

\def\ve{{\bm{e}}}

\def\vq{{\bm{q}}}

\def\evlambda{{\lambda}}

\def\evg{{g}}

\def\mE{{\bm{E}}}

\def\mG{{\bm{G}}}
\def\mH{{\bm{H}}}

\def\mP{{\bm{P}}}
\def\mQ{{\bm{Q}}}
\def\mR{{\bm{R}}}
\def\mS{{\bm{S}}}

\def\mW{{\bm{W}}}
\def\mX{{\bm{X}}}

\def\mZ{{\bm{Z}}}

\DeclareMathAlphabet{\mathsfit}{\encodingdefault}{\sfdefault}{m}{sl}
\SetMathAlphabet{\mathsfit}{bold}{\encodingdefault}{\sfdefault}{bx}{n}
\newcommand{\tens}[1]{\bm{\mathsfit{#1}}}

\def\tX{{\tens{X}}}

\def\gA{{\mathcal{A}}}
\def\gB{{\mathcal{B}}}

\def\gN{{\mathcal{N}}}

\def\gP{{\mathcal{P}}}

\def\gV{{\mathcal{V}}}

\def\sI{{\mathbb{I}}}
\def\sJ{{\mathbb{J}}}

\def\sO{{\mathbb{O}}}

\def\sU{{\mathbb{U}}}
\def\sV{{\mathbb{V}}}

\def\emG{{G}}

\def\emR{{R}}
\def\emS{{S}}

\def\emZ{{Z}}

\newcommand{\pmodel}{p_{\rm{model}}}

\newcommand{\Ls}{\mathcal{L}}
\newcommand{\R}{\mathbb{R}}

\newcommand{\softmax}{\mathrm{softmax}}

\def\vlambda{{\bm{\lambda}}}

\def\Appref#1{Appendix~\ref{#1}}

\def\Tabref#1{Table~\ref{#1}}

\usepackage{amsthm}
\newtheorem{proposition}{Proposition}

\usepackage{algorithm}
\usepackage{algorithmic}
\usepackage{enumitem}
\usepackage{needspace}
\usepackage{graphicx}
\usepackage{wrapfig}
\usepackage{subcaption}
\usepackage{booktabs}
\usepackage{multirow}
\usepackage{makecell}
\usepackage{xcolor}
\usepackage{placeins}
\usepackage{lipsum}

\usepackage{xcolor}

\title{Pruned CTC for Memory-Efficient Large-\\Vocabulary ASR Training}

\author{%
\textbf{%
Yifan Yang\textsuperscript{1}\thanks{%
Work done during internship at Alibaba Token Hub, Alibaba Group.}\phantom{$^{*}$},
Xiaoyu Yang\textsuperscript{4},
Zengrui Jin\textsuperscript{3},
Xian Shi\textsuperscript{2},
Yuxuan Wang\textsuperscript{2},
Yu Xi\textsuperscript{2}}\\
\textbf{%
Ziyang Ma\textsuperscript{1},
Qi Chen\textsuperscript{1},
Ruiyang Xu\textsuperscript{1},
Hui Wang\textsuperscript{5},
Dongchao Yang\textsuperscript{6},
Jin Xu\textsuperscript{2}\thanks{Project leader. \textsuperscript{\ensuremath{\ddagger}}Corresponding author.}\phantom{$^{\dagger}$},
Xie Chen\textsuperscript{1,7}\footnotemark[3]}\\
\textsuperscript{1}Shanghai Jiao Tong University,
\textsuperscript{2}Alibaba Token Hub, Alibaba Group,
\textsuperscript{3}Tsinghua University\\
\textsuperscript{4}University of Cambridge,
\textsuperscript{5}Nankai University,
\textsuperscript{6}Chinese University of Hong Kong,
\textsuperscript{7}SII\\
{\hypersetup{hidelinks}\href{https://github.com/yfyeung/PrunedCTC}{\faGithub\enspace https://github.com/yfyeung/PrunedCTC}}
}

\begin{document}

\maketitle

\begin{abstract}
Connectionist temporal classification (CTC) naturally supports offline and streaming speech recognition with utterance-level supervision, but conventional implementations materialize frame-by-vocabulary activations in memory, making CTC training with native LLM vocabularies prohibitively memory-intensive. A key observation is that every valid CTC alignment uses only target tokens and blank, and their union across a batch typically forms a small subset of the full vocabulary. We introduce Pruned CTC, which restricts alignment computation to this subset while retaining full-vocabulary normalization. We prove that this vocabulary reduction is exactly equivalent to full-vocabulary CTC in loss and gradients. Head-and-loss activation memory no longer scales linearly with vocabulary size. We further apply finite-beam alignment pruning.  Building on Pruned CTC, we develop LLM-CTC, which adapts pretrained LLMs for non-autoregressive ASR while retaining causal attention and native vocabularies, and extend it to bounded-history streaming, avoiding chunk-level speech--text alignments. Experiments show that, with Zipformer-M encoder and 180K vocabulary, Pruned CTC reduces full-step memory by 5.1$\times$ with only 17\% step-time overhead. Across three corpora, it matches standard CTC accuracy. On GigaSpeech, across six Qwen3 model sizes from 0.6B to 32B, LLM-CTC remains within 7\% relative WER of LLM-CE with 7 to 10$\times$ faster recognition; when fine-tuning Qwen3-ASR for bounded-history streaming, LLM-CTC remains within 3\% relative WER of matched offline models on the test set. Together, these results establish Pruned CTC as a scalable sequence objective for native-vocabulary LLM ASR across offline and streaming settings.
\end{abstract}
\section{Introduction}
\label{sec:intro}

Automatic speech recognition (ASR) systems increasingly use pretrained large language models (LLMs) to draw on their linguistic knowledge~\citep{ma2024slam, shi2026qwen3asr}.
However, LLM-based ASR typically uses token-level cross-entropy (CE) training and autoregressive generation, which we call LLM-CE, requiring sequential decoding.
Streaming further requires coordinating text emission with incoming speech, often relying on chunk-level speech--text alignments for training~\citep{seide2024speechreallm,jia2024speechllmxl}.

Connectionist temporal classification (CTC)~\citep{graves2006ctc} provides a natural alternative, supporting offline and streaming recognition with utterance-level transcripts.
However, standard implementations materialize frame-by-vocabulary activations in memory, making CTC training costly with native LLM vocabularies.
Note that every valid CTC alignment contains only blank and the target tokens, while all other vocabulary classes contribute only through the softmax normalizer.
We therefore introduce Pruned CTC, which restricts alignment computation to the target classes present in a batch along with blank while retaining full-vocabulary normalization.
Pruned CTC's head-and-loss activation memory no longer scales linearly with vocabulary size.
We prove that this reduction is exactly equivalent to full-vocabulary CTC in loss and first-order gradients, and further apply finite-beam alignment pruning.

Making large-vocabulary CTC practical allows us to revisit how pretrained LLMs are used for ASR. 
Building on Pruned CTC, we propose LLM-CTC to adapt pretrained LLMs for non-autoregressive ASR over their native vocabularies.
This design extends to bounded-history streaming recognition under utterance-level supervision, avoiding chunk-level speech--text alignments during training.

In summary, our contributions to the community include:
\begin{itemize}[leftmargin=*, itemsep=0pt, topsep=2pt]
\vspace{-0.5em}

\item A memory-efficient CTC training method, Pruned CTC, with exact loss and first-order gradient equivalence under vocabulary reduction and head-and-loss activation memory that no longer scales linearly with vocabulary size, further combined with finite-beam alignment pruning.

\item A CTC-based LLM adaptation method, LLM-CTC, that enables non-autoregressive ASR with pretrained causal LLMs while retaining causal attention and native vocabularies.
Across six Qwen3 model sizes from 0.6B to 32B, LLM-CTC remains within 7\% relative WER of LLM-CE with 7 to 10$\times$ faster recognition.

\item A streaming extension of LLM-CTC that avoids chunk-level speech--text alignments during training and enables bounded-history inference with key-value (KV) cache reuse.

% \item A forthcoming open-source implementation suite covering Pruned CTC, LLM-CTC, and streaming LLM-CTC training and inference, to facilitate further research and applications in ASR.

\end{itemize}

\section{Related Work}
\label{sec:related}

\paragraph{Memory-efficient loss.}
Pruned RNN-T~\citep{kuang2022prunedrnnt} restricts the time--label region the full joiner evaluates while keeping the vocabulary dimension unchanged.
CoLaCTC~\citep{zhang2023colactc} coarsens the CTC label inventory for auxiliary speech-translation training.
Cut Cross-Entropy~\citep{wijmans2025cce} avoids storing full-vocabulary logits for tokenwise CE, where the target class is known at each position.
CTC lacks frame-level targets: label posteriors are computed by forward--backward over the transcript.
These posteriors are supported only on transcript labels and blank, while full-vocabulary normalization yields dense gradients.
Pruned CTC exploits this structure by separating reduced-vocabulary alignment computation from full-vocabulary normalization and gradients, using chunked recomputation to avoid materializing frame-by-vocabulary activations.
Finite-beam pruning separately limits the alignment sum.

\paragraph{CTC with large language models.}
LegoSLM~\citep{ma2025legoslm} uses speech-encoder CTC posteriors over an LLM vocabulary as weighted inputs for autoregressive decoding.
\citet{baskar2024speech} add auxiliary CTC to speech-prefix outputs under bidirectional prefix attention and evaluate CTC and autoregressive decoding.
\citet{schmitt2026llmspeech} report reduced recognition accuracy from auxiliary CTC on pretrained LLM speech-prefix outputs.
With CTC as the main objective,~\citet{syu2024ctcpmlm} fine-tune bidirectional pretrained language models and their vocabulary heads for non-autoregressive translation.
NLE~\citep{dekel2026nle} converts a pretrained causal LLM into a bidirectional editor conditioned on speech and an initial CTC hypothesis with insertion slots.
Ablations retain causal attention or replace the initial hypothesis with equal-length blank sequences.

\paragraph{Streaming LLM-based ASR.}
Speech ReaLLM~\citep{seide2024speechreallm} interleaves speech and text using external CTC alignments for CE training, including adaptation of a pretrained 7B LLM.
SpeechLLM-XL~\citep{jia2024speechllmxl} assigns transcript segments to audio chunks using alignments and limits cached history.
BESTOW~\citep{chen2024bestow} uses a fixed read/write policy, while~\citet{wan2026streamingllm} learn emission timing with monotonic chunkwise attention.

\section{Pruned CTC}
\label{sec:method}

\subsection{Preliminary: Connectionist Temporal Classification}
\label{sec:method:prelim}

\paragraph{Formulation.} Speech recognition maps a sequence of encoder frames to a sequence of target labels, typically with more frames than labels.
CTC~\citep{graves2006ctc} augments the label vocabulary with a blank symbol, yielding an output vocabulary $\sV$ of size $V = |\sV|$.
CTC assigns to a blank-free target sequence $y$ the total probability of all frame-level alignments $\pi$ that collapse to $y$.
Let $\gB$ denote the many-to-one collapse map that first merges consecutive repetitions and then removes blank symbols.
The set of length-$t$ alignments consistent with $y$ is $\gB^{-1}_t(y) = \{\pi \in \sV^{t} : \gB(\pi) = y\}$.

A batch holds $N$ utterances, with positive frame count $T_n$ and blank-free target $y_n$ for utterance $n$.
Let $T=\max_n T_n$ and $M = \sum_n T_n \le NT$.
Stacking the $M$ frames row-wise gives $\mX \in \R^{M \times D}$, with $D$ the head input width and $\sI_n$ the row indices of utterance $n$ in frame order. All class column indices start at zero, sequence positions start at one, and slices include both endpoints.
Let $\gN_+ \subseteq \{1, \dots, N\}$ index the utterances whose alignment set $\gB^{-1}_{T_n}(y_n)$ is nonempty.
A vocabulary head, an affine map with weight $\mW \in \R^{V \times D}$ and bias $\vb \in \R^{V}$, turns them into logits $\mZ = \mX \mW^{\top} + \vone_M \vb^{\top} \in \R^{M \times V}$, where $\vone_M$ is the all-ones vector of length $M$, so that one column of $\mZ$ is one class. A $\softmax$ normalizes row $m$ into a per-frame class distribution, or in log space
\begin{equation}
  \log p_v(m) = \emZ_{m,v} - \evlambda_m ,
  \qquad
  \evlambda_m = \log \sum_{u \in \sV} \exp \emZ_{m,u} .
  \label{eq:lognorm}
\end{equation}
We call $\evlambda_m$ the normalizer, which depends on all class scores at frame $m$.
CTC assumes conditionally independent frame-level labels given the input sequence, so the batch loss is
\vspace{-0.25em}
\begin{equation}
  \Ls = -\sum_{n \in \gN_+} \log \sum_{\pi \in \gB^{-1}_{T_n}(y_n)} \prod_{t=1}^{T_n} p_{\pi_t}(\sI_n[t]) .
  \label{eq:ctc}
\vspace{-0.5em}
\end{equation}

\paragraph{Discussion.} A standard implementation evaluates \Eqref{eq:ctc} by materializing $\mZ$, normalizing it along the class axis, and passing the log-probabilities to a dynamic program. A trainable vocabulary head requires $\Theta(VD)$ storage for parameters, gradients and optimizer state, and $\Theta(MVD)$ projection arithmetic. Full-vocabulary logits and their gradients require additional activation storage.
\begin{itemize}[leftmargin=*, itemsep=0pt, topsep=2pt]
\vspace{-0.25em}
  \item The forward and backward passes store several $M \times V$ arrays. At large $V$, these arrays can exceed the encoder's activation storage and dominate training memory.
  \item The dynamic program uses only blank and the labels appearing in the transcripts of the batch. When these classes make up a small fraction of the vocabulary, most columns are materialized solely for the normalizer $\evlambda_m$.
  \item Removing classes with positive probability mass changes the normalizer and hence the CTC objective. The reduced input must retain full-vocabulary normalization.
\end{itemize}

\subsection{Pruned CTC via vocabulary reduction}
\label{sec:method:exact}

\paragraph{Reformulation.} We propose \emph{Pruned CTC}, a memory-efficient CTC training method that separates alignment computation for the target transcript from normalization over the full vocabulary.
The selected classes $\sU_n$ are blank and the distinct labels of $y_n$.
For $n\in\gN_+$, all valid alignments visit each frame once and share the same normalization factors, giving the utterance loss $\Ls_n$:
\begin{equation}
  \Ls_n = \sum_{t=1}^{T_n}\log\sum_{v\in\sV}\exp\emZ_{\sI_n[t],v}
  - \log\sum_{\pi\in\gB^{-1}_{T_n}(y_n)}
    \exp\!\left(\sum_{t=1}^{T_n}\emZ_{\sI_n[t],\pi_t}\right).
  \label{eq:ctc-factorization}
\end{equation}
The first term normalizes over the entire vocabulary. Every valid alignment satisfies $\pi_t\in\sU_n$, so the second term uses only logits for these classes. Classes outside $\sU_n$ still affect the loss and receive gradients through the first term. Vocabulary reduction, chunked projection and dense gradient reconstruction evaluate \Eqref{eq:ctc} without retaining the full frame-by-vocabulary logit matrix.

Let $\sU = \bigcup_n \sU_n$ be the selected classes for the batch, with $K = |\sU|$, and $\sO = \sV \setminus \sU$ its complement. To preserve the normalization term in \Eqref{eq:ctc-factorization}, we pool the nonblank classes in $\sO$ into the \emph{other class}.
Order $\sU$ with blank first and relabel each target token by its index in $\sU$.
When $\sO$ is nonempty, define $\mR \in \R^{M \times (K+1)}$ by
\begin{equation}
  \emR_{m,j} = \emZ_{m,\, \sU[j]} \quad (j = 0, \dots, K-1),
  \qquad
  \emR_{m,K} = \log \sum_{v \in \sO} \exp \emZ_{m,v},
  \label{eq:grouped}
\end{equation}
whose first and last columns represent blank and the complement's mass, respectively.
When $\sO$ is empty, the grouped logits $\mR$ have only the $K$ columns indexed by $\sU$.
The other class is nonblank and absent from every target, so its symbol survives collapsing and cannot appear in valid alignments.
Pooling preserves the normalizer, $\log\sum_j\exp\emR_{m,j}=\evlambda_m$, so each class in $\sU$ keeps its probability.

\begin{proposition}
\label{prop:exact}
In exact arithmetic, CTC on $\log\softmax(\mR)$ with relabeled targets and all valid alignments, summed over the same utterances in $\gN_+$, equals $\Ls$ as a function of $\mZ$. The reduced and full-vocabulary losses therefore have identical first-order gradients with respect to $\mZ$, $\mX$, $\mW$ and $\vb$.
\end{proposition}

We call this computation with vocabulary reduction and all valid alignments \emph{unpruned reduced CTC}.

For utterance $n\in\gN_+$, the occupancy $\gamma_v(m)$ is the posterior probability of label $v$ at frame $m$, conditioned on $y_n$.
\Appref{app:proof} proves the proposition and derives the logit gradient:
\begin{equation}
  \frac{\partial \Ls}{\partial \emZ_{m,v}} = p_v(m) - \gamma_v(m)
  \qquad \text{for all } v \in \sV, \ m \in \sI_n, \ n \in \gN_+ .
  \label{eq:grad-full}
\end{equation}
Labels outside $\sU_n$ have zero occupancy and retain gradient $p_v(m)$, so the head gradient remains dense over the vocabulary. Frames of an utterance outside $\gN_+$ have zero gradient.

\subsection{Finite-beam approximation}
\label{sec:method:beam}

Our implementation additionally prunes the alignment lattice using k2~\citep{k2}.
The alignment beam $\beta$ sets a log-score tolerance relative to the best alignment.
\Appref{app:impl:beam} details this pruning rule and bounds the error in logit gradients.
Let $\gP_{\beta,n}\subseteq\gB^{-1}_{T_n}(y_n)$ be the retained alignments for utterance $n$.
When every utterance in $\gN_+$ has a nonempty retained lattice, we denote its loss by $\Ls_{\beta,n}$ and use the batch objective
\begin{equation}
  \Ls_\beta = -\sum_{n\in\gN_+}\log
  \sum_{\pi\in\gP_{\beta,n}}\prod_{t=1}^{T_n}p_{\pi_t}(\sI_n[t]).
  \label{eq:pruned-ctc}
\end{equation}
Vocabulary reduction preserves every retained alignment's probability. Restricting the alignment sum gives $\Ls_\beta\geq\Ls$, with equality when all valid alignments are retained.
With the retained lattice locally fixed, frames $m\in\sI_n$ for $n\in\gN_+$ have logit gradient $p_v(m)-\gamma_v^\beta(m)$, where $\gamma_v^\beta(m)$ is the posterior occupancy over retained alignments.
Excluded utterances have zero gradient.
Pruning boundaries can introduce discontinuities in the loss.

All Pruned CTC training runs use an alignment beam of 100.
Across random initializations of encoder-based and LLM-based ASR in \Appref{app:impl:beam}, the estimated discarded alignment posterior mass does not exceed $10^{-11}$ per utterance at this beam.

\subsection{Forward and backward computation}
\label{sec:method:alg}

\noindent
\begin{minipage}[t]{0.474\textwidth}
\vspace{0pt}
\parskip=.5pc
\Algref{alg:pruned-ctc} projects vocabulary chunks of at most $C$ columns to compute $\Ls_\beta$ when every utterance in $\gN_+$ has a finite retained-lattice loss.
Online log-sum-exp updates $\vlambda\in\R^M$, one full-vocabulary normalizer per frame.
Each chunk is released after use.
Projecting onto the selected classes and subtracting $\vlambda$ gives their log-probabilities $\mS=(\log\softmax(\mR))_{:,0:K-1}$.
The complement contributes through $\vlambda$, so the dynamic program uses only these $K$ columns.
The backward pass recomputes each chunk, forms gradients, and releases it.
\Appref{app:impl:backend} details target relabeling and padding.

\paragraph{Recomputed dense gradient.} With $\mG = \partial \Ls_\beta / \partial \mS$ from the dynamic program and row sums $\evg_m = \sum_{j<K} \emG_{m,j}$,
\begin{equation}
  \frac{\partial \Ls_\beta}{\partial \emZ_{m,v}}
  = \operatorname{scatter}_{\sU}(\mG)_{m,v} - \evg_m\,p_v(m).
  \label{eq:vjp}
\end{equation}
\end{minipage}
\hfill
\begin{minipage}[t]{0.51\textwidth}
\setlength{\intextsep}{0pt}
\begin{algorithm}[H]
\caption{Pruned CTC}
\label{alg:pruned-ctc}
\small
\begin{algorithmic}[1]
  \STATE \textbf{input:} padded head-input states
         $\tX \in \R^{N \times T \times D}$,\\
         targets $y_{1:N}$, frame counts $T_{1:N}$, head $(\mW, \vb)$,\\
         chunk width $C$, alignment beam $\beta$
  \STATE $\sU \gets$ blank followed by sorted distinct labels in $y_{1:N}$%
  \STATE $K \gets |\sU|$, \ $\tilde y_{1:N} \gets y_{1:N}$ relabeled by $\sU$
  \STATE $\mX \gets$ the $M = \sum_n T_n$ valid frames of $\tX$
  \STATE $\vlambda \gets$ online log-sum-exp over
         vocabulary chunks $\mZ_{:,\, v_0:v_1-1}$
  \STATE $\mS \gets \mX \mW_{\sU}^{\top} + \vone_M \vb_{\sU}^{\top}
         - \vlambda \vone_K^{\top}$
  \STATE $\gA \gets \textsc{Intersect}(\textsc{CtcGraph}(\tilde y_{1:N}), \mS, \beta)$
  \STATE $\Ls_{\beta,1:N} \gets -\textsc{TotalScores}(\gA;\,\text{log semiring})$
  \STATE zero nonfinite entries of $\Ls_{\beta,1:N}$
  \STATE \textbf{output:} $\Ls_\beta = \sum_n \Ls_{\beta,n}$ and
         $\mG = \partial \Ls_\beta / \partial \mS$
  \STATE \textbf{backward:} $\evg_m \gets \sum_{j<K} \emG_{m,j}$
  \FOR{each vocabulary chunk $v_0:v_1-1$}
    \STATE recompute $\mZ_{:,\, v_0:v_1-1}$
    \STATE form \Eqref{eq:vjp} for this chunk
    \STATE accumulate into $\partial \Ls_\beta / \partial \mX$
    \STATE write its rows of $\partial \Ls_\beta / \partial \mW$ and
           $\partial \Ls_\beta / \partial \vb$
  \ENDFOR
\end{algorithmic}
\end{algorithm}
\end{minipage}

The first term places $\mG$ in columns $\sU$ and zeros elsewhere, giving at most $K$ nonzero columns.
The second uses the full-vocabulary probabilities of \Eqref{eq:lognorm} and supplies normalization gradients to all columns, including the complement.
For a unit-weight loss on a nonempty retained lattice, $\emG_{m,j}=-\gamma_{\sU[j]}^\beta(m)$ and $-\evg_m=1$.
With all valid alignments retained, \Eqref{eq:vjp} recovers \Eqref{eq:grad-full}.
The actual row sums also account for loss scaling and excluded utterances, as described in \Appref{app:impl:grad}.
Recomputing and consuming one dense chunk at a time reduces class-axis activation storage from $\Theta(MV)$ to $\Theta(MK + M)$ plus a temporary block of at most $M \times C$ entries.
Storage for the dynamic program also depends on sequence lengths and the alignment beam.

When the two gradient terms nearly cancel for the selected classes, rounding errors can dominate their difference. \Appref{app:impl:precision} compares precision strategies and quantifies loss and gradient errors across vocabulary and batch sizes.

\section{LLM-CTC}
\label{sec:ctcllm}

\subsection{Preliminary: LLM-based ASR}
\label{sec:ctcllm:prelim}

\paragraph{Formulation.} A speech encoder and projector map utterance $n$ to $\mE_n\in\R^{T_n\times D}$. A pretrained decoder-only language model receives these speech embeddings and a $P$-token prompt $\mP\in\R^{P\times D}$. The model's causal transformer backbone $f$, including the final normalization, produces one $D$-dimensional state per position.
The pretrained vocabulary projection, or LLM head, uses the weight $\mW$ from \Secref{sec:method:prelim} and a pretrained bias $\vb^{\mathrm{pre}}\in\R^V$. For autoregressive recognition, $\vb=\vb^{\mathrm{pre}}$.

Adding the end-of-sequence token gives $y_n^+=(y_n,\mathrm{EOS})$.
Autoregressive (AR) training minimizes
\begin{equation}
  \Ls_{\mathrm{CE}} =
  -\sum_n \sum_{i=1}^{|y_n|+1}
    \log \pmodel\!\left(y^+_{n,i}\mid\mP,\mE_n,y^+_{n,<i}\right).
  \label{eq:ce}
\end{equation}
LLM-CE generates one token per sequential LLM forward pass, conditioned on preceding text, and reuses a KV cache across forward passes.

\begin{figure}[t]
\centering
\includegraphics[width=\linewidth]{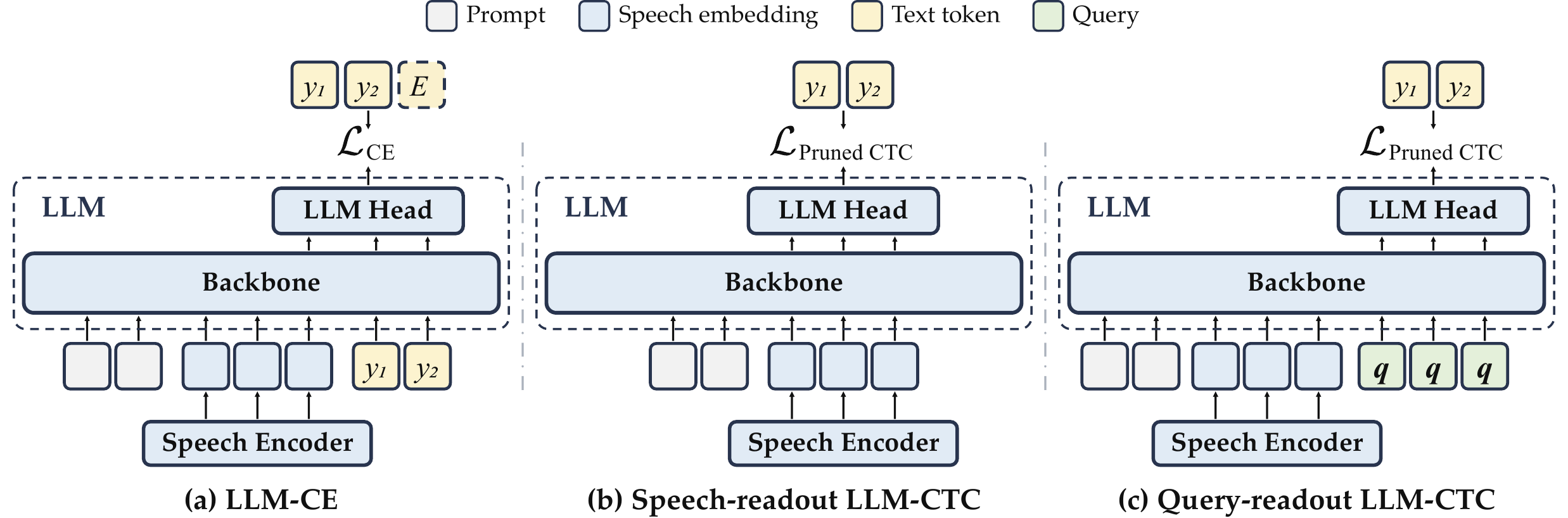}
\caption{Offline LLM-based ASR under causal attention.
\captiona{}~LLM-CE predicts transcript tokens and EOS ($E$).
\captionb{}~Speech-readout LLM-CTC reads speech positions with prefix access.
\captionc{}~Query-readout LLM-CTC reads appended queries with access to all speech embeddings.
$\Ls_{\mathrm{CE}}$ is cross-entropy, and $\Ls_{\text{Pruned CTC}}=\Ls_\beta$ is finite-beam CTC with alignment beam $\beta$.}
\label{fig:ctc-llm-arch}
\end{figure}

\subsection{Offline LLM-CTC}
\label{sec:ctcllm:readouts}
\label{sec:ctcllm:query}
\label{sec:ctcllm:context}

\paragraph{Formulation.} LLM-CTC adapts pretrained LLMs for direct, non-autoregressive speech recognition over their native vocabularies using Pruned CTC. Offline systems process the entire utterance before decoding. We retain causal attention to preserve the pretrained attention structure and facilitate streaming with KV cache reuse.
\Figref{fig:ctc-llm-arch} compares LLM-CE with the two LLM-CTC readouts, which select different LLM positions for CTC prediction.

\emph{Speech-readout LLM-CTC} processes the prompt followed by speech embeddings to obtain $\mH_n^{\mathrm{s}}=f([\mP^\top\,\mE_n^\top]^\top)$ and selects $\mX_{\sI_n,:}=(\mH_n^{\mathrm{s}})_{P+1:P+T_n,:}$ as the CTC states.
Each speech position can attend to the prompt and the speech positions up to itself.

\emph{Query-readout LLM-CTC} appends $T_n$ copies of one learned vector $\vq\in\R^D$ after the speech embeddings, computes $\mH_n^{\mathrm{q}}=f([\mP^\top\,\mE_n^\top\,\vq\vone_{T_n}^\top]^\top)$ and selects $\mX_{\sI_n,:}=(\mH_n^{\mathrm{q}})_{P+T_n+1:P+2T_n,:}$ as the CTC states.
The copies occupy distinct sequence positions. Every query can attend to the prompt, all speech embeddings and query positions up to itself.

Both variants provide $T_n$ CTC states for utterance $n$, stacked into the batch matrix $\mX$.
An output class $v_{\mathrm{bl}}$ unused by the tokenizer serves as blank, preserving the text vocabulary.
Since the LLM head is frozen, we learn a blank offset $\eta$ to adjust the blank logit: $\vb=\vb^{\mathrm{pre}}+\eta\ve_{v_{\mathrm{bl}}}$, where $\ve_{v_{\mathrm{bl}}}\in\R^V$ is the standard basis vector for blank.
Full-vocabulary normalization of the LLM head logits gives $p_v(m)$, including blank, as defined in \Eqref{eq:lognorm}. Both variants minimize the finite-beam objective $\Ls_\beta$ in \Eqref{eq:pruned-ctc} using these probabilities.

\paragraph{Discussion.} Both variants support non-autoregressive (NAR) recognition, producing $T_n$ CTC frames in $\Theta(1)$ LLM forward passes and admitting the same valid alignments for a given transcript.
\begin{itemize}[leftmargin=*, itemsep=0pt, topsep=2pt]
\item Speech embeddings already carry encoder context. Within the LLM, Query-readout LLM-CTC gives every CTC prediction position direct access to the speech embeddings of the entire utterance.
Speech-readout LLM-CTC restricts this access to a growing speech prefix.
\item Query-readout LLM-CTC adds $D$ parameters and increases the LLM input length from $P+T_n$ to $P+2T_n$, increasing tokenwise and attention work in the same LLM layers. \Appref{app:readout:cost} gives the corresponding cost analysis.
\end{itemize}

\subsection{Streaming LLM-CTC}
\label{sec:ctcllm:streaming}

\begin{figure}[t]
\centering
\includegraphics[width=\linewidth]{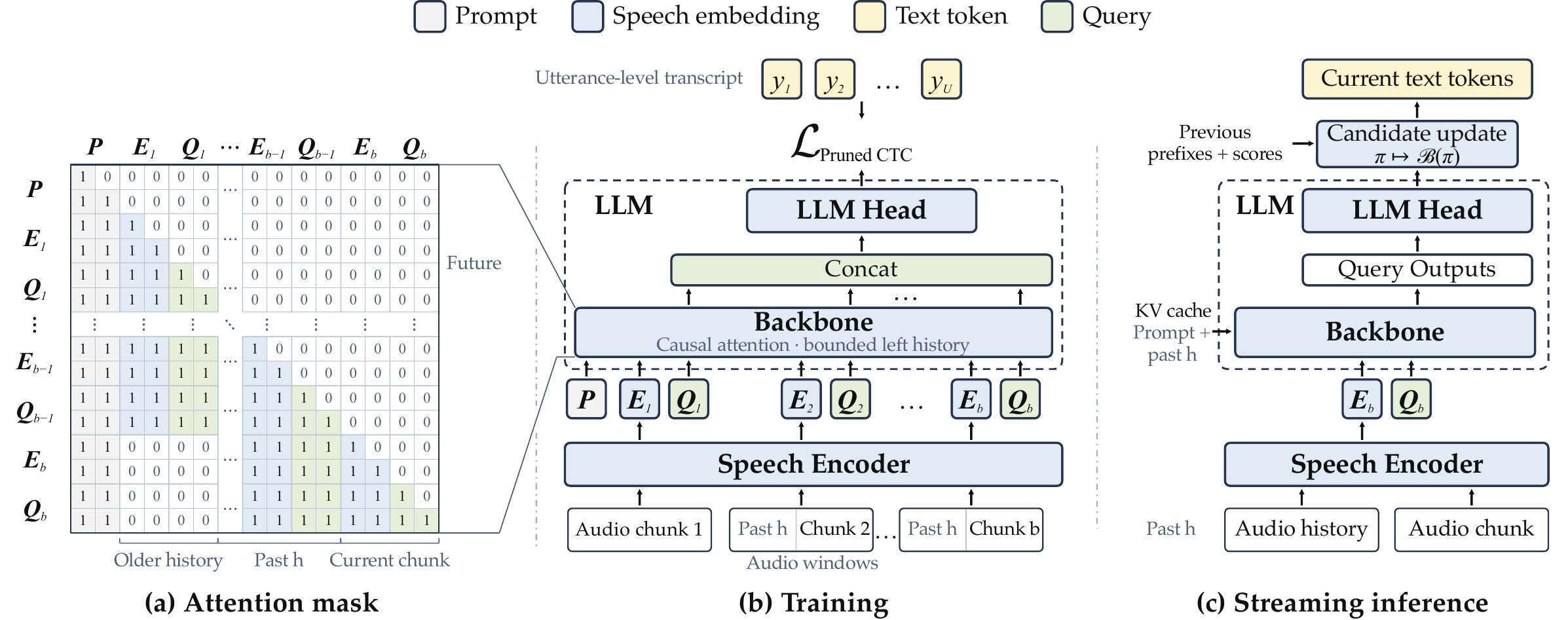}
\caption{Streaming Query-readout LLM-CTC.
\captiona{}~Chunk positions attend to prompt $\mP$, $h$ seconds of speech/query history before the chunk, and their causal prefixes within the chunk.
\captionb{}~Speech blocks $\mE_b$ alternate with shared-query blocks $\mQ_b$.
Concatenated query outputs are trained against the entire $U$-token transcript using $\Ls_{\text{Pruned CTC}}=\Ls_\beta$.
\captionc{}~Each chunk uses one LLM forward pass with a KV cache, while CTC prefixes and scores persist across chunks.}
\label{fig:streaming-readout}
\end{figure}

\paragraph{Formulation.} \emph{Streaming Query-readout LLM-CTC} extends the causal query readout by processing utterance $n$ in $B_n$ chunks. For chunk $b$, the encoder uses $h$ seconds of left history and returns only current-chunk outputs. The projector maps them to speech embeddings $\mE_{n,b}\in\R^{T_{n,b}\times D}$. We append the shared-query block $\mQ_{n,b}=\vone_{T_{n,b}}\vq^\top$.
Let $f_h$ denote $f$ with the mask in \Figref{fig:streaming-readout}\captiona{} at every LLM layer.
A position in chunk $b$ can attend to the prompt, speech/query positions from the $h$ seconds before that chunk, and positions up to itself within the chunk.
The LLM outputs $\mH_n^{\mathrm{str}} = f_h([\mP^\top\,\mE_{n,1}^\top\,\mQ_{n,1}^\top\,\cdots\, \mE_{n,B_n}^\top\,\mQ_{n,B_n}^\top]^\top)$.
Each CTC prediction position can directly attend to the entire current speech block.
Feature extraction and encoding exclude audio from later chunks.

\paragraph{Training.} We sample the left-history length $h$ per microbatch as specified in \Appref{app:config:qwen-asr} and use it for the speech encoder and LLM attention mask.
With $T_n = \sum_{b=1}^{B_n} T_{n,b}$ and sequence offset $o_{n,b} = P + 2\sum_{c=1}^{b-1}T_{n,c}$, we concatenate query outputs in chunk order as in \Figref{fig:streaming-readout}\captionb{}:
\begin{equation}
  \mX_{\sI_n,:}
    = \operatorname{Concat}_{b=1}^{B_n}
    (\mH_n^{\mathrm{str}})_{o_{n,b}+T_{n,b}+1:o_{n,b}+2T_{n,b},:} .
  \label{eq:streaming-readout}
\end{equation}
The frozen LLM head and blank offset $\eta$ produce frame probabilities for $\Ls_\beta$.
Applied once to the entire transcript, this loss sums over retained alignments across chunks, avoiding chunk-level speech--text alignments.
Training can process all chunks in parallel under this attention rule.

\paragraph{Inference.} Each chunk produces its query emissions in one LLM forward pass, as in \Figref{fig:streaming-readout}\captionc{}.
The speech encoder retains at most $h$ seconds of left history between chunks.
The LLM reuses the KV cache for the prompt and retained speech/query positions, discarding older entries.
CTC prefix beam search carries candidate prefixes and scores across chunks.
The best CTC hypothesis may change as new emissions arrive, while earlier LLM emissions remain fixed.
\Appref{app:readout:streaming} proves the equivalence of parallel batch computation and per-utterance cached computation for the LLM in exact arithmetic.
\Appref{app:readout:consistency} quantifies their finite-precision differences.

\section{Experiments}
\label{sec:exp}

\subsection{Memory and runtime}
\label{sec:exp:scaling}

\paragraph{Encoder-based ASR.} We benchmark Pruned CTC with a Zipformer-M encoder on GigaSpeech using one H100 80G GPU. Standard CTC uses a dense vocabulary head followed by PyTorch's CTC loss.
We vary vocabulary size from 500 to 180,000 classes and total batch frames $M$ over 16,000, 32,000 and 64,000, retokenizing transcripts for each vocabulary. Both methods process the same batches.
We measure forward and backward through the vocabulary head and CTC loss, and full training steps including the speech encoder and optimizer.
\Appref{app:config:bench} gives the settings.

\begin{figure}[t]
\centering
\begin{minipage}[t]{0.55\textwidth}
\vspace{0pt}
\centering
\includegraphics[width=\linewidth]{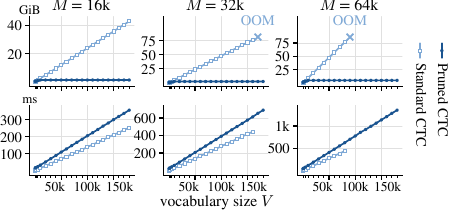}
\captionsetup{position=bottom,margin=0pt,skip=2pt}
\caption{Vocabulary head and CTC loss memory (top) and time (bottom) for forward and backward passes.
Crosses: extrapolated out-of-memory (OOM) points.}
\label{fig:vocab-scaling}
\end{minipage}
\hfill
\begin{minipage}[t]{0.43\textwidth}
\vspace{0pt}
\centering
\captionsetup{type=table,position=top,margin=0pt}
\caption{Full training steps for each batch size at the largest tested vocabulary feasible for standard CTC.
Memory is in GiB.
Memory reductions and runtime overheads are relative to standard CTC.}
\label{tab:whole-step}
\footnotesize
\setlength{\tabcolsep}{1.2pt}
\renewcommand{\arraystretch}{1.12}
\begin{tabular}{@{}rrcccc@{}}
\toprule
\multirow{2}{*}{$M$} & \multirow{2}{*}{$V$}
& \multicolumn{3}{c}{Memory}
& \multirow{2}{*}{\makecell{Runtime\\overhead}} \\
\cmidrule(lr){3-5}
& & Standard & Pruned & Reduction & \\
\midrule
16k & 180k & 51.7 & 10.1 & 5.1$\times$ & +17\% \\
32k & 130k & 77.3 & 17.9 & 4.3$\times$ & +17\% \\
64k & 50k & 76.2 & 33.8 & 2.3$\times$ & +14\% \\
\bottomrule
\end{tabular}
\end{minipage}
\end{figure}

\begin{itemize}[leftmargin=*, itemsep=0pt, topsep=2pt]
\item \textbf{Head and loss.} \Figref{fig:vocab-scaling} shows approximately constant memory for Pruned CTC over 10,000 to 180,000 classes at each batch size, while standard CTC's memory grows linearly with vocabulary size.
At $M=16{,}000$ and $V=180{,}000$, Pruned CTC reduces memory from 43.0 to 1.29 GiB, a factor of 33.3. At the same $V$, it uses 2.58 and 5.16 GiB for $M=32{,}000$ and $M=64{,}000$, respectively.
Memory savings begin between 2,000 and 5,000 classes at all three batch sizes. Pruned CTC completes all tested head-and-loss configurations, including those where standard CTC runs out of memory.
Pruned CTC's runtime ratio to standard CTC generally decreases with vocabulary size. At $M=16{,}000$, the ratio is 9.14 at 500 classes and 1.41 at 180,000 classes.

\item \textbf{Full training steps.} \Tabref{tab:whole-step} uses the largest tested vocabulary feasible for standard CTC at each batch size. Pruned CTC reduces memory by a factor of 2.3 to 5.1, with 14 to 17\% more time per step.
At $V = 180{,}000$, Pruned CTC also completes full steps at $M = 32{,}000$ and $64{,}000$, using 18.2 and 34.5 GiB, respectively.
Standard CTC runs out of memory in both settings.
\end{itemize}

\paragraph{LLM-based ASR.} We profile the 8B and 14B LLM-CTC models in \Secref{sec:exp:llm} on one A100 80G GPU, using LoRA and frozen native LLM heads with $V=151{,}936$.
We compare standard CTC, unpruned reduced CTC and Pruned CTC under the settings in \Appref{app:config:bench:llm}.

\begin{itemize}[leftmargin=*, itemsep=0pt, topsep=2pt]
\item \textbf{Head and loss.} Memory is the peak increase above resident head parameters and inputs during forward and backward.
Relative to standard CTC, Pruned CTC reduces this from 29.21 to 1.14 GiB for 8B and from 29.79 to 1.20 GiB for 14B, by factors of 25.7 and 24.8, respectively.

\item \textbf{Full training steps.} Including resident model and optimizer storage, Pruned CTC reduces peak memory from 56.06 to 34.67 GiB for 8B and from 72.71 to 52.84 GiB for 14B relative to standard CTC.
These are reductions of 38.2\% and 27.3\%, with 7.5\% and 5.6\% more time per step, respectively. Storage and activations outside the head and loss limit the full-step reductions.
Unpruned reduced CTC achieves nearly the same full-step memory and runtime as Pruned CTC. Exact vocabulary reduction and recomputation provide most of the measured memory savings.
\end{itemize}

\FloatBarrier
\subsection{Encoder-based ASR}
\label{sec:exp:nn}

\paragraph{Setup.} 
We compare Pruned CTC with standard CTC on three Mandarin and English corpora.
LibriSpeech~\citep{panayotov2015librispeech} is 960 hours of read English, GigaSpeech~\citep{chen2021gigaspeech} is 10,000 hours of read and spontaneous English speech, and AISHELL-1~\citep{bu2017aishell} is 170 hours of read Mandarin. LibriSpeech and GigaSpeech use byte-pair encoding (BPE)~\citep{kudo2018sentencepiece} with 500-class vocabularies.
AISHELL-1 uses a 4336-class character vocabulary. We use a Zipformer-M encoder~\citep{yao2024zipformer} and follow the standard CTC configurations in CR-CTC~\citep{yao2025crctc}, replacing the loss with Pruned CTC.
Decoding uses prefix beam search with beam size 4.
We report word error rate (WER) for English and character error rate (CER) for Mandarin.
\Appref{app:config:nn} details training configurations.

\begin{table}[H]
\centering
\caption{ASR performance comparison between standard CTC and Pruned CTC with a Zipformer-M encoder: WER (\%) on LibriSpeech and GigaSpeech, and CER (\%) on AISHELL-1.}
\label{tab:wer}
\begin{tabular}{lcccccc}
\toprule[1pt]
\multirow{2}{*}{Method}
      & \multicolumn{2}{c}{LibriSpeech} & \multicolumn{2}{c}{GigaSpeech} & \multicolumn{2}{c}{AISHELL-1} \\
\cmidrule(lr){2-3} \cmidrule(lr){4-5} \cmidrule(lr){6-7}
      & \textit{test-clean} & \textit{test-other} & \textit{dev} & \textit{test} & \textit{dev} & \textit{test} \\
\midrule
Standard CTC~\citep{yao2025crctc} & 2.52 & 6.02 & 11.23 & 11.27 & 4.47 & 4.80 \\
Pruned CTC (ours)                 & 2.52 & 6.02 & 11.25 & 11.24 & 4.49 & 4.79 \\
\bottomrule[1pt]
\end{tabular}
\vspace{-1em}
\end{table}

\paragraph{Results.} Pruned CTC yields nearly identical WER and CER to standard CTC across all three corpora in \Tabref{tab:wer}. These results show that Pruned CTC with finite-beam alignment pruning can replace standard CTC in the evaluated ASR recipes without sacrificing recognition accuracy.

\subsection{Offline LLM-based ASR}
\label{sec:exp:llm}

\paragraph{Setup.}
On GigaSpeech, we pair a frozen SPEAR-XLarge v2 speech encoder~\citep{yang2026spear} with Qwen3 LLMs~\citep{yang2025qwen3} and refer to this model family as SPEAR-Qwen3.
We evaluate Qwen3 backbones at 0.6B, 1.7B, 4B, 8B, 14B, and 32B parameters.
Paired LLM-CE and LLM-CTC models share the encoder and Qwen3 size.
WER uses beam search for LLM-CE and prefix beam search for LLM-CTC, both with beam size 4.
Real-time factor (RTF) is total timed recognition time divided by total audio duration.
We measure it on the trained checkpoints used for WER, with greedy decoding on one H100 80G GPU.
\Appref{app:config:rtf} details the timing protocol.
\Appref{app:config:llm} gives shared training settings, and \Appref{app:config:spear} specifies the models.
Additional baselines include Whisper~\citep{radford2022whisper}, Qwen2-Audio~\citep{chu2024qwen2audio}, and Qwen3-ASR~\citep{shi2026qwen3asr}, all decoded with beam size 4 as detailed in \Appref{app:config:baselines}.

\begingroup
\setlength{\intextsep}{0pt}
\setlength{\columnsep}{10pt}
\begin{wraptable}[11]{r}{0.42\textwidth}
\vspace{-1.4pt}
\begin{minipage}[t]{\linewidth}
\vspace{0pt}
\centering
\setlength{\tabcolsep}{2.8pt}
\captionsetup{position=top,skip=4pt}
\caption{Readout and attention ablations on GigaSpeech, using SPEAR-\allowbreak Qwen3-\allowbreak 4B.
WER (\%) and RTF.}
\label{tab:llm-readout}
\begin{tabular}{@{}llccc@{}}
\toprule
\multirow{2}{*}{Readout} & \multirow{2}{*}{Attention} & \multicolumn{2}{c}{GigaSpeech} & \multirow{2}{*}{RTF} \\
\cmidrule(lr){3-4}
& & \textit{dev} & \textit{test} & \\
\midrule
Speech & Causal & 10.24 & 10.46 & 0.0108 \\
Speech & Full & 10.01 & 10.16 & 0.0114 \\
Query & Causal & 10.05 & 10.21 & 0.0109 \\
\bottomrule
\end{tabular}
\end{minipage}
\end{wraptable}

\paragraph{Readout and attention.}
\Tabref{tab:llm-readout} compares SPEAR-\allowbreak Qwen3-\allowbreak 4B readouts under matched initialization, training and decoding settings, with one CTC frame per speech embedding.
Under causal attention, Query-readout LLM-CTC has lower WER on both splits at similar RTF. The full-attention ablation makes prompt and speech positions mutually visible during training and decoding.
Speech-readout LLM-CTC with full attention gives slightly lower WER at higher RTF, but bidirectional attention over the entire utterance must become causal across chunks for streaming with KV cache reuse. We use Query-readout LLM-CTC to access the entire utterance offline or the current speech block in streaming with causal attention in the LLM.
\par
\endgroup

\paragraph{Accuracy and speed.}
\Tabref{tab:llm-systems} shows that Query-readout LLM-CTC benefits from larger Qwen3 backbones.
With the speech encoder and native LLM head frozen, test WER decreases from 10.57\% at 0.6B to 9.94\% at 32B, while remaining within 7\% relative WER of LLM-CE across all six sizes.
The 4B, 8B, 14B and 32B LLM-CTC systems also achieve lower test WER than SPEAR-XLarge CTC and RNN-T, suggesting that non-autoregressive ASR can benefit from linguistic knowledge in pretrained LLMs.
LLM-CE also provides a strong autoregressive baseline, achieving test WERs of 9.43\% and 9.41\% at 14B and 32B, respectively.
The 14B and 32B LLM-CE models outperform the strong baseline Qwen3-ASR-1.7B, whose audio encoder was pretrained on approximately 40M hours of ASR data~\citep{shi2026qwen3asr}.
LLM-CTC reduces RTF by 7.2 to 10.3$\times$ relative to LLM-CE, reflecting the removal of sequential AR generation.

\begin{table}[t]
\centering
\caption{ASR performance comparison between LLM-CE and Query-readout LLM-CTC using SPEAR-Qwen3 on GigaSpeech. WER (\%) and RTF are reported.
Gray baselines use ASR training data not restricted to GigaSpeech.
WER uses official GigaSpeech normalization, with num2words(en) and punctuation removal first applied to hypotheses from gray baselines.}
\label{tab:llm-systems}
\setlength{\tabcolsep}{2.8pt}
\begin{tabular}{@{}l*{2}{>{\centering\arraybackslash}p{2.4em}}c@{\hspace{10pt}}*{2}{>{\centering\arraybackslash}p{2.4em}}c@{\hspace{18pt}}l*{2}{>{\centering\arraybackslash}p{2.4em}}@{}}
\cmidrule[\heavyrulewidth](r{18pt}){1-7}\cmidrule[\heavyrulewidth]{8-10}
& \multicolumn{3}{@{\hspace{\cmidrulekern}}c@{\hspace{\cmidrulekern}}}{\textbf{\textit{LLM-CE}}}
& \multicolumn{3}{@{\hspace{\cmidrulekern}}c@{\hspace{18pt}}}{\textbf{\textit{LLM-CTC}}}
& \multicolumn{3}{@{}c@{}}{\textbf{\textit{Baselines}}} \\
\cmidrule(lr){2-4}\cmidrule(l{5pt}r{18pt}){5-7}\cmidrule{8-10}
\multirow{2}{*}{Size} & \multicolumn{2}{c}{GigaSpeech} & \multirow{2}{*}{RTF}
& \multicolumn{2}{c}{GigaSpeech} & \multirow{2}{*}{RTF}
& \multirow{2}{*}{Model} & \multicolumn{2}{c}{GigaSpeech} \\
\cmidrule(lr){2-3}\cmidrule(lr){5-6}\cmidrule(l){9-10}
& \textit{dev} & \textit{test} & & \textit{dev} & \textit{test} &
& & \textit{dev} & \textit{test} \\
\cmidrule[\lightrulewidth](r{18pt}){1-7}\cmidrule[\lightrulewidth]{8-10}
0.6B & 9.81 & 9.88 & 0.0724 & 10.43 & 10.57 & 0.0099
& SPEAR-XLarge CTC & 10.30 & 10.41 \\
1.7B & 9.60 & 9.78 & 0.0732 & 10.12 & 10.34 & 0.0102
& SPEAR-XLarge RNN-T & 10.09 & 10.26 \\
4B & 9.56 & 9.65 & 0.0944 & 10.05 & 10.21 & 0.0109
& \textcolor{gray}{Whisper large-v3} & \textcolor{gray}{11.75} & \textcolor{gray}{11.25} \\
8B & 9.56 & 9.59 & 0.0963 & 9.90 & 10.11 & 0.0111
& \textcolor{gray}{Qwen2-Audio-7B} & \textcolor{gray}{11.14} & \textcolor{gray}{11.06} \\
14B & 9.34 & 9.43 & 0.1014 & 9.78 & 10.03 & 0.0121
& \textcolor{gray}{Qwen3-ASR-0.6B} & \textcolor{gray}{9.70} & \textcolor{gray}{9.66} \\
32B & 9.35 & 9.41 & 0.1546 & 9.73 & 9.94 & 0.0150
& \textcolor{gray}{Qwen3-ASR-1.7B} & \textcolor{gray}{9.53} & \textcolor{gray}{9.62} \\
\cmidrule[\heavyrulewidth](r{18pt}){1-7}\cmidrule[\heavyrulewidth]{8-10}
\end{tabular}
\vspace{-1em}
\end{table}

\subsection{Streaming LLM-based ASR}
\label{sec:exp:streaming}

\begingroup
\setlength{\intextsep}{0pt}
\setlength{\columnsep}{10pt}
\begin{wraptable}{r}{0.48\textwidth}
\vspace{-1.4pt}
\begin{minipage}[t]{\linewidth}
\vspace{0pt}
\captionsetup{position=top,skip=4pt}
\centering
\caption{Offline and streaming Query-readout LLM-CTC with fine-tuned Qwen3-ASR.
WER (\%) and left-history length $h$ (seconds).}
\label{tab:streaming}
\small
\setlength{\tabcolsep}{2.8pt}
\begin{tabular}{@{}lc*{4}{>{\centering\arraybackslash}p{2.9em}}@{}}
\toprule
\multirow{3}{*}{Mode} & \multirow{3}{*}{$h$}
& \multicolumn{2}{c}{0.6B} & \multicolumn{2}{c}{1.7B} \\
\cmidrule(lr){3-4}\cmidrule(lr){5-6}
& & \multicolumn{2}{c}{GigaSpeech} & \multicolumn{2}{c}{GigaSpeech} \\
\cmidrule(l{8pt}r{8pt}){3-4}\cmidrule(l{8pt}r{8pt}){5-6}
& & \textit{dev} & \textit{test} & \textit{dev} & \textit{test} \\
\midrule
Offline & \textemdash{} & 10.69 & 10.73 & 10.00 & 10.16 \\
\midrule
\multirow{3}{*}{Streaming} & 2 & 10.98 & 11.02 & 10.30 & 10.35 \\
& 4 & 10.96 & 11.01 & 10.32 & 10.33 \\
& 8 & 10.95 & 10.99 & 10.32 & 10.32 \\
\bottomrule
\end{tabular}
\end{minipage}
\end{wraptable}

\paragraph{Setup.} On GigaSpeech, we separately fine-tune Qwen3-ASR~\citep{shi2026qwen3asr} 0.6B and 1.7B for offline and streaming Query-readout LLM-CTC.
Both use Pruned CTC with utterance-level transcripts.
LoRA fine-tuning keeps the speech encoder and LLM head frozen.
Streaming follows \Secref{sec:ctcllm:streaming} with 2-second chunks and a shared left-history length $h$ for the encoder and LLM. We vary $h$ during fine-tuning and evaluate each streaming model with 2, 4, and 8 seconds of left history.
Both modes use prefix beam search with beam size 4.
See \Appref{app:config:qwen-asr} for details.

\paragraph{Results.} \Tabref{tab:streaming} shows that Query-readout LLM-CTC retains near-offline accuracy in bounded-history streaming.
On \textit{test}, streaming increases WER by less than 3\% relative to matched offline models across both model sizes and all evaluated left-history lengths.
Increasing left history from 2 to 8 seconds changes WER little, indicating limited benefit from longer history.
\par
\endgroup

\section{Conclusion}
\label{sec:conclusion}

We present Pruned CTC, a memory-efficient CTC training method combining exact vocabulary reduction and full-vocabulary normalization with finite-beam alignment pruning.
Compared to standard CTC, its head-and-loss activation memory no longer scales linearly with vocabulary size, while matching standard CTC accuracy.
Building on Pruned CTC, we propose LLM-CTC to adapt pretrained LLMs for non-autoregressive ASR while retaining causal attention and native vocabularies.
Across six Qwen3 sizes, LLM-CTC remains within 7\% relative WER of LLM-CE with 7 to 10$\times$ faster recognition.
The streaming extension of LLM-CTC pairs utterance-level supervision with bounded-history inference and KV cache reuse, staying within 3\% relative WER of matched offline models on the GigaSpeech test set.
Pruned CTC makes native-vocabulary CTC a practical basis for adapting pretrained LLMs to both offline and streaming ASR.

\clearpage

\subsection*{AI use statement}
In this work, we used generative AI tools only to aid and polish writing, specifically to correct grammar and improve the clarity of passages drafted by the authors. We have not used generative AI tools for research ideation, method design, code implementation, experiments, result analysis, or the literature review. We have reviewed all AI-assisted work and take responsibility for the final content of this work, including text, claims or artifacts produced with the aid of generative AI.

\subsection*{Reproducibility statement}
The datasets used in our experiments are publicly available.
The algorithm is specified in \Secref{sec:method:alg}, implementation and precision details in \Appref{app:impl}, and the experimental setup and model configurations in \Secref{sec:exp}, \Appref{app:resources} and \Appref{app:config}.
To facilitate the reproduction of our results, the source code and pretrained checkpoints will be publicly available.

\bibliography{iclr2027_conference}
\bibliographystyle{iclr2027_conference}

\clearpage

\appendix

\section{Proof of Proposition~\ref{prop:exact}}
\label{app:proof}

Let $m$ index the stacked frames and $j$ the grouped columns in \Eqref{eq:grouped}.
We derive gradients on frames of utterances in $\gN_+$, whose alignment sets are nonempty.
Utterances outside this set contribute zero loss and gradient under \Eqref{eq:ctc}.
With the normalizer $\evlambda_m$ of \Eqref{eq:lognorm}, write
\begin{equation*}
  \bar p_j(m) = \softmax(\mR_{m,\cdot})_j
  \qquad \text{and} \qquad
  P_{\sO}(m) = \sum_{v \in \sO} p_v(m),
\end{equation*}
for the grouped class distribution and for the probability mass assigned to the complement $\sO$. If $\sO$ is empty the other column is dropped, $\mR$ is a permutation of the columns of $\mZ$ and the claim is immediate, so take $\sO$ to be nonempty below.

\begin{proof}
We first show that grouping preserves the CTC loss and then verify equality of the gradients.

\paragraph{Step 1: valid alignments use only $\sU_n$.}
Every alignment in $\gB^{-1}_{T_n}(y_n)$ emits only blank and labels of $y_n$.
A nonblank symbol absent from $y_n$ would survive collapsing, so it cannot appear in such an alignment. This excludes the other class and gives
\begin{equation*}
  \gamma_v(m) \equiv 0
  \qquad \text{for all } m \in \sI_n \text{ and } v \in \sV \setminus \sU_n .
\end{equation*}
Labels contributed by other utterances therefore add no valid alignments for utterance $n$.
Relabeling by the indices in $\sU$ preserves blank positions and label equality, so it gives a bijection between the original and grouped valid alignments.

\paragraph{Step 2: the other column reproduces the normalizer.}
Log-sum-exp is associative over any partition of the class axis, and $\{\sU, \sO\}$ partitions $\sV$.
With $\mR$ as defined in \Eqref{eq:grouped},
\begin{equation*}
  \log \sum_{j=0}^{K} \exp \emR_{m,j}
  \; = \; \log \Big( \sum_{v \in \sU} \exp \emZ_{m,v} \; + \sum_{v \in \sO} \exp \emZ_{m,v} \Big)
  \; = \; \log \sum_{v \in \sV} \exp \emZ_{m,v}
  \; = \; \evlambda_m .
\end{equation*}
The first $K$ columns of $\log \softmax(\mR)$ are consequently
\begin{equation*}
  \emR_{m,j} - \evlambda_m = \emZ_{m,\, \sU[j]} - \evlambda_m = \log p_{\sU[j]}(m) ,
\end{equation*}
the full-vocabulary normalized log-probabilities themselves, so $\bar p_j(m) = p_{\sU[j]}(m)$ for $j < K$ and $\bar p_K(m) = P_{\sO}(m)$. Every valid alignment from Step 1 therefore keeps its original probability. The utterance loss in \Eqref{eq:ctc-factorization} is unchanged by grouping.
The alignment posteriors and frame occupancies $\gamma_v(m)$ are therefore unchanged under relabeling.
The objectives agree as functions of $\mZ$. Finite logits give a positive alignment sum for each alignable utterance, so their gradients agree.

\paragraph{Step 3: gradient equivalence.}
Treating the log-probabilities as independent inputs to the CTC recurrence, differentiating the negative log of the sum of alignment probabilities gives
\begin{equation*}
  \frac{\partial \Ls}{\partial \log \bar p_j(m)} = -\gamma_{\sU[j]}(m)
  \qquad (j < K) ,
\end{equation*}
and zero on the other column. Every alignment emits one label at each frame, so the posterior occupancies at that frame satisfy $\sum_{v\in\sV}\gamma_v(m)=1$. Applying the log-softmax Jacobian gives
\begin{equation}
  \frac{\partial \Ls}{\partial \emR_{m,j}}
  = -\gamma_{\sU[j]}(m) + \bar p_j(m) \sum_{v \in \sV} \gamma_v(m)
  = \bar p_j(m) - \gamma_{\sU[j]}(m) ,
  \qquad
  \frac{\partial \Ls}{\partial \emR_{m,K}} = P_{\sO}(m) ,
  \label{eq:grad-grouped}
\end{equation}
since the other column has zero alignment occupancy. For a selected class at column $j < K$, the map $\emZ_{m,\sU[j]} \mapsto \emR_{m,j}$ is the identity, so Step 2 turns the first half of \Eqref{eq:grad-grouped} into $p_v(m) - \gamma_v(m)$ at $v = \sU[j]$. On a complementary column, the chain rule through \Eqref{eq:grouped} gives
\begin{equation*}
  \frac{\partial \emR_{m,K}}{\partial \emZ_{m,v}} = \exp(\emZ_{m,v} - \emR_{m,K}) ,
\end{equation*}
so the explicit log-sum-exp over $\sO$ cancels:
\begin{equation}
  \begin{aligned}
    \frac{\partial \Ls}{\partial \emZ_{m,v}}
    &= P_{\sO}(m) \, \exp(\emZ_{m,v} - \emR_{m,K}) \\[2pt]
    &= \exp(\emR_{m,K} - \evlambda_m) \, \exp(\emZ_{m,v} - \emR_{m,K})
     \; = \; p_v(m) ,
    \qquad v \in \sO .
  \end{aligned}
  \label{eq:offsupport}
\end{equation}
Since $\gamma_v(m) \equiv 0$ for $v \in \sO$, the right-hand side is also $p_v(m) - \gamma_v(m)$, so both cases recover \Eqref{eq:grad-full}. By the chain rule, the gradients with respect to $\mX$, $\mW$ and $\vb$ are $(\partial \Ls / \partial \mZ)\mW$, $(\partial \Ls / \partial \mZ)^{\top}\mX$ and $(\partial \Ls / \partial \mZ)^{\top}\vone_M$, respectively. Since both formulations use the same vocabulary head and give the same logit gradient, these three gradients also agree.
\end{proof}

\section{Pruned CTC Implementation and Validation}
\label{app:impl}

\subsection{Loss and gradient computation}
\label{app:impl:backend}
\label{app:impl:grad}

\paragraph{Forward computation.} We stack the $M$ head-input states without padding before applying the vocabulary head.
The selected classes $\sU$ are ordered with blank first, followed by sorted distinct target labels. Targets are relabeled with these indices, with blank at index zero for k2~\citep{k2}.
The finite-beam CTC computation evaluates $\Ls_\beta$ from the $K$ columns of $\mS$ and each utterance's starting row and frame count. The full-vocabulary normalizer is accumulated and subtracted in 64-bit floating point (FP64) before the entries of $\mS$ are rounded once to 32-bit floating point (FP32).
For unpruned reduced CTC, we retain all valid alignments and apply class-axis log-softmax to grouped logits $\mR$. We then restore padding for PyTorch CTC, using frame counts to exclude padded frames.

\paragraph{Loss aggregation.} A target with $r_n$ pairs of identical adjacent labels requires at least $|y_n|+r_n$ frames. This condition defines the structurally alignable utterances $\gN_+$ in \Eqref{eq:ctc}. The finite-beam computation sums finite per-utterance losses and replaces nonfinite losses with zero. This selects the same utterances as \Eqref{eq:ctc} when every structurally alignable utterance has a finite loss. Lattice capacity limits can narrow the beam or leave an empty lattice. Numerical failures can also make a loss nonfinite.

\paragraph{Backward computation.} With the retained lattice locally fixed, differentiate $\emS_{m,j}=\emZ_{m,\sU[j]}-\evlambda_m$ using $\partial\evlambda_m/\partial\emZ_{m,v}=p_v(m)$. The chain rule gives
\begin{equation*}
  \frac{\partial\Ls_\beta}{\partial\emZ_{m,v}}
  = \sum_{j:\sU[j]=v}\emG_{m,j}
    - \left(\sum_{j<K}\emG_{m,j}\right)p_v(m),
\end{equation*}
where the first sum is zero outside $\sU$, yielding \Eqref{eq:vjp}.
The row sums of $\mG$ account for loss scaling and excluded utterances. Sparse and dense terms are combined before matrix products and sums. Each vocabulary chunk computes distinct rows of the head gradients and adds its contribution to the head-input gradient. For unpruned reduced CTC, gradients pass through log-softmax on $\mR$, distributing the other-column gradient to complement classes through log-sum-exp, as in \Eqref{eq:offsupport}. This computation also proceeds in chunks.
Finite-beam CTC stores $M$ normalizers and $MK$ log-probabilities for the selected classes, with a temporary FP64 normalization block of at most $M\times C$ entries. Each device stores the full head.

\subsection{Numerical accuracy}
\label{app:impl:precision}

We compare Pruned CTC with a dense reference that computes vocabulary projection, normalization and unpruned CTC in FP64. Paired evaluations share encoder outputs, head parameters and targets. Relative error is the norm of the difference divided by the reference norm, using Frobenius norms for matrices and Euclidean norms for vectors.
Absolute error is the largest entrywise difference. Gradients use the sum of utterance losses.

\paragraph{Effect of precision and score reuse.}
We use NVIDIA-H100-SXM5-80GB GPUs to compare precision choices when $p_v(m)-\gamma_v^\beta(m)$ can nearly cancel.
We use a 500-class CTC model trained for 50 epochs, with a Zipformer-M encoder that produces 512-dimensional outputs. Five GigaSpeech \textit{dev} batches contain 41,368 to 58,903 frames and 477 to 486 selected classes. Projections and CTC inputs use FP32, with alignment beam 100 and chunk width 4096. We vary forward normalization precision, reuse of $\mS$ during the backward pass, and the precision of the matrix products and sums that form gradients.
With FP64 forward normalization, reuse preserves the subtraction in FP64 before rounding to FP32. Recomputation subtracts an FP32 copy of the normalizer.

\begin{table}[t]
\centering
\caption{Precision choices after training. Entries are gradient error ratios for the selected classes relative to the first configuration's error on each batch. Ranges span five batches.}
\label{tab:precision-ablation}
\centering
\footnotesize
\setlength{\tabcolsep}{4.5pt}
\begin{tabular}{@{}ccccc@{}}
\toprule
\multirow{2}{*}{\makecell{Forward\\normalization}}
& \multirow{2}{*}{\makecell{Log-probabilities for\\selected classes in backward}}
& \multirow{2}{*}{\makecell{Gradient products\\and sums}}
& \multicolumn{2}{c}{Gradient error ratio} \\
\cmidrule(lr){4-5}
& & & Weight & Bias \\
\midrule
FP32 & recomputed & FP32 & 1.000 & 1.000 \\
FP32 & reused & FP32 & 1.000 & [0.998, 1.000] \\
FP64 & recomputed & FP32 & [0.348, 0.474] & [0.229, 0.357] \\
FP64 & reused & FP32 & [0.061, 0.141] & [0.062, 0.095] \\
FP64 & reused & FP64 & [0.047, 0.087] & [0.062, 0.097] \\
\bottomrule
\end{tabular}
\end{table}

\begin{figure}[t]
\centering
\includegraphics[width=\linewidth]{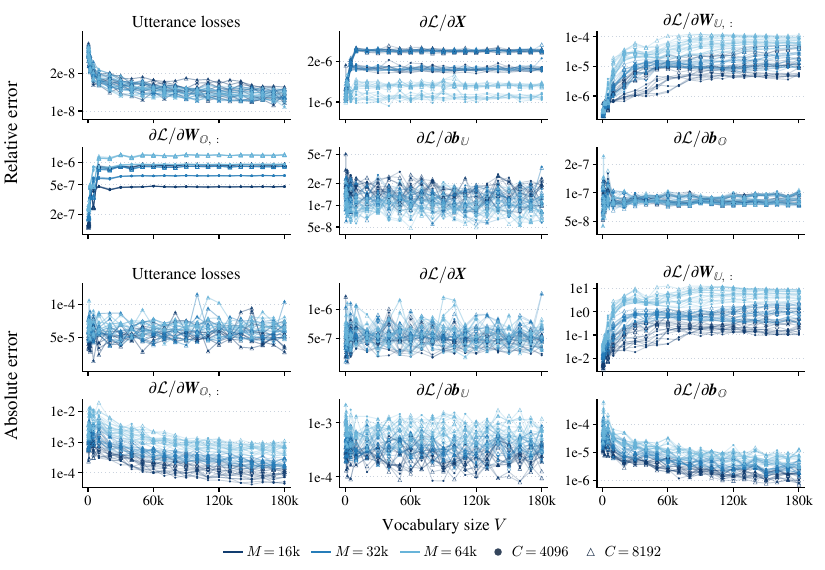}
\captionsetup{position=bottom,skip=5pt}
\caption{Numerical error across vocabulary sizes: relative error above and absolute error below.
Every panel shows all 1,188 comparisons with a linear x-axis and logarithmic y-axis.
Each curve tracks one initialization and batch at fixed $M$ and $C$.}
\label{fig:ctc-precision}
\end{figure}

\Tabref{tab:precision-ablation} supports the combined precision strategy used in Pruned CTC: FP64 normalization, reuse of $\mS$ and FP32 gradient products and sums.
Relative weight-gradient errors for the selected classes range from $1.18\times10^{-6}$ to $3.32\times10^{-6}$.
FP64 gradient products improve weight accuracy further, with mixed effects on bias gradients.

\paragraph{Effect of vocabulary and batch size.}
We measure loss and gradient accuracy across vocabulary and batch sizes on NVIDIA-H100-SXM5-80GB GPUs, using randomly initialized CTC models with Zipformer-M encoders.
We use 22 vocabularies: 500, 1,000, 2,000 and 5,000 classes, followed by 10,000 to 180,000 in increments of 10,000. Batch frame counts are 16,000, 32,000 and 64,000. Each vocabulary and frame-count pair uses three initializations, three GigaSpeech training batches per initialization and chunk widths 4096 and 8192, giving 1,188 comparisons. Encoder inference uses bfloat16 (BF16) mixed precision, with outputs converted to FP32 for the vocabulary head. The dense reference groups entire utterances to bound memory, while the reduced computation processes full batches.
Both use unpruned FP64 CTC recurrences to isolate class-axis arithmetic error.

\Figref{fig:ctc-precision} shows every comparison, with loss errors measured over the vector of utterance losses. Maximum relative errors are $3.45\times10^{-8}$ for loss, $2.65\times10^{-6}$ for head-input gradients and $1.21\times10^{-4}$ for weight gradients for the selected classes. The latter persists when both implementations receive identical gradients with respect to $\mS$. On the worst batch, with $M=64{,}000$ and $V=80{,}000$, FP64 gradient products reduce the relative weight-gradient error to $9.94\times10^{-8}$, showing that FP32 accumulation is the main source of error.

\subsection{Alignment pruning}
\label{app:impl:beam}

Each CTC alignment traverses the lattice from its start state to a final state.
k2's dense intersection routine prunes states and arcs using the score of the best alignment passing through each one, relative to the best alignment overall. The CTC loss computation sums over all retained alignments. Retained arcs can recombine into alignments scoring more than $\beta$ below the best alignment.

\begin{samepage}
\paragraph{Discarded posterior mass.} For $n\in\gN_+$, use identical emission scores and a nonempty retained lattice. The loss gap $\Delta_n=\Ls_{\beta,n}-\Ls_n\geq0$ satisfies
\begin{equation*}
  \exp(-\Delta_n)
  = \frac{\displaystyle\sum_{\pi\in\gP_{\beta,n}}
      \prod_{t=1}^{T_n}p_{\pi_t}(\sI_n[t])}
    {\displaystyle\sum_{\pi\in\gB^{-1}_{T_n}(y_n)}
      \prod_{t=1}^{T_n}p_{\pi_t}(\sI_n[t])}
  = 1-\delta_n .
\end{equation*}
Thus $\delta_n=1-\exp(-\Delta_n)$ is the posterior mass of discarded alignments, conditioned on the target transcript. When $\delta_n>0$, let $\gamma_v^{\mathrm{drop}}(m)$ be the occupancy conditioned on discarded alignments. The full occupancy is the mixture
\begin{equation*}
  \gamma_v(m)=(1-\delta_n)\gamma_v^\beta(m)
    +\delta_n\gamma_v^{\mathrm{drop}}(m),
  \qquad m\in\sI_n .
\end{equation*}
Each occupancy vector is a probability distribution, so in exact arithmetic
\begin{equation*}
  |\gamma_v^\beta(m)-\gamma_v(m)|\leq\delta_n,
  \qquad
  \sum_{v\in\sV}|\gamma_v^\beta(m)-\gamma_v(m)|\leq2\delta_n .
\end{equation*}
When $\delta_n=0$, the occupancies agree. For a locally fixed retained lattice, the shared $p_v(m)$ terms cancel, giving the same bounds on differences between unit-weight logit gradients.
\end{samepage}

Numerical estimates $\widehat{\delta}_n$ include rounding from both dynamic programs and are not certified upper bounds.
A batch meets the criterion when $\max_n\widehat{\delta}_n\leq10^{-11}$.
The required alignment beam is the smallest tested value meeting this criterion.
We retain signed estimates, including zeros.

\paragraph{Encoder-based ASR.} We evaluate CTC models with Zipformer-M encoders at random initialization on NVIDIA-H100-SXM5-80GB GPUs.
We use shared FP32 scores, promoted to FP64 for the unpruned recurrence.
Both losses and $\widehat{\Delta}_n=\widehat{\Ls}_{\beta,n}-\widehat{\Ls}_n$ are computed in FP64, with $\widehat{\delta}_n=-\operatorname{expm1}(-\widehat{\Delta}_n)$.

We first evaluate a 500-class model on ten GigaSpeech training batches containing 6,090 utterances at three random initializations and under eight independent augmentations with MUSAN~\citep{snyder2015musan} and SpecAugment~\citep{park2019specaugment}. These combinations yield 240 batch evaluations and 146,160 utterance evaluations at each of 18 alignment beams from 10 to 150. Initializations share audio batches within each augmentation.

Across augmentations, the largest required alignment beam ranges from 80 to 90.
All batches meet the criterion at $\beta=90$, ten log-score units below the training value of 100, as shown in \Figref{fig:beam-sweep}. The largest $\widehat{\delta}_n$ at $\beta=150$ is about $2.3\times10^{-12}$.
The log plot shows positive upper-tail statistics.

We then evaluate a grid of 22 vocabularies: 500, 1,000, 2,000 and 5,000 classes, followed by 10,000 to 180,000 in increments of 10,000. Batch frame counts are 16,000, 32,000 and 64,000. Batches contain 160, 320 or 640 utterances, each with exactly 100 encoder frames. Each vocabulary and batch-size pair is evaluated at three random initializations, using three batches per initialization.
For each of these 198 configurations, we report the largest required alignment beam across its three batches, for 594 batch evaluations in total.

All requirements lie between 40 and 70 in \Figref{fig:beam-vocab}. The dashed line marks the training alignment beam of 100, which meets the measured criterion across both the variable-length training batches and this vocabulary grid at the sampled initializations.
A separate numerical-accuracy check over these vocabulary sizes and frame counts uses three random initializations and three batches per initialization. At $C=4096$ and $\beta=100$, its 594 paired comparisons with unpruned CTC use identical emission scores and give a maximum absolute loss difference of $6.82\times10^{-13}$ per utterance.

\begin{figure}[t]
\centering
\begin{minipage}[t]{0.485\textwidth}
\vspace{0pt}
\centering
\includegraphics{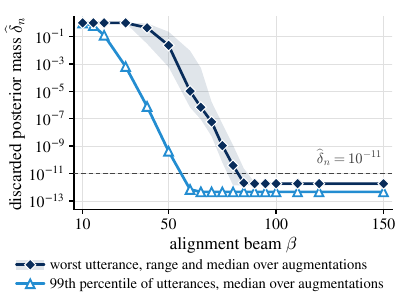}
\caption{Estimated discarded posterior mass on GigaSpeech over eight augmentations, with three initializations each.}
\label{fig:beam-sweep}
\end{minipage}
\hfill
\begin{minipage}[t]{0.485\textwidth}
\vspace{0pt}
\centering
\includegraphics{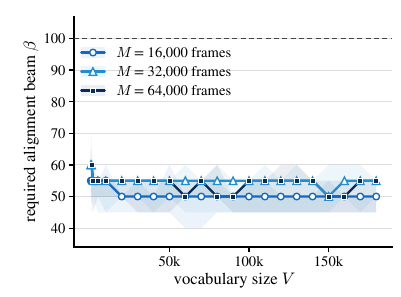}
\caption{Required alignment beam for $\widehat{\delta}_n\leq10^{-11}$ per utterance on GigaSpeech: median and range over initializations.}
\label{fig:beam-vocab}
\end{minipage}
\end{figure}

\paragraph{LLM-based ASR.} We evaluate offline Query-readout LLM-CTC using the SPEAR-Qwen3-8B architecture, with $V=151{,}936$.
All components are initialized from configuration without pretrained weights. No parameters are updated.
We use 25 GigaSpeech training batches containing 1,383 utterances at three random initializations, giving 4,149 evaluations per beam.
All initializations and beams share features augmented with MUSAN noise and SpecAugment.

Independent recurrences in extended precision compute the full and retained alignment sums on identical FP32 scores with shared per-frame scaling.
We obtain $\widehat{\delta}_n$ as one minus the ratio of retained to full alignment sums.

\Tabref{tab:llm-pruning} shows that 75 is the smallest tested alignment beam meeting the $10^{-11}$ criterion for all evaluations.
Across all tested beams, no empty retained lattices or nonfinite losses occur before zeroing.
Estimates at $\beta\geq100$ are at the numerical floor of an unpruned lattice control and do not establish exact zero discarded mass.

\begin{table}[htbp]
\centering
\caption{Discarded posterior mass for randomly initialized SPEAR-Qwen3-8B.}
\label{tab:llm-pruning}
\begin{tabular}{rcc}
\toprule
Alignment beam $\beta$ & Maximum $\widehat{\delta}_n$ & Above $10^{-11}$ \\
\midrule
10 & 0.999858 & 2,847 / 4,149 \\
25 & 0.749547 & 1,301 / 4,149 \\
50 & $2.09\times10^{-4}$ & 14 / 4,149 \\
75 & $1.42\times10^{-12}$ & 0 / 4,149 \\
100, 150, 200 & Numerical floor & 0 / 4,149 \\
\bottomrule
\end{tabular}
\end{table}

\section{Resource Benchmarks}
\label{app:resources}

\subsection{Shared settings}
\label{app:resource:measurement}

\paragraph{Data.} The resource benchmarks use GigaSpeech training audio and the same augmentation methods as training: MUSAN~\citep{snyder2015musan} noise and SpecAugment.

\paragraph{Measurement.} We measure forward and backward through the vocabulary head and CTC loss, and full training steps.
For the head and loss, memory is the largest increase above resident head parameters and inputs after clearing their gradients.
Full-step memory is the largest allocated amount during the complete model forward pass, CTC loss, backward, the optimizer update and gradient reset, including resident model and optimizer storage.
Elapsed wall-clock time is measured with device synchronization at both boundaries. Memory uses PyTorch's allocated-memory counters.

\subsection{Encoder-based ASR}
\label{app:config:bench}

\paragraph{Data and vocabularies.} We use 80-dimensional FBank features.
Each utterance contributes exactly 100 encoder frames without frame padding, giving $M/100$ utterances per batch. The three batch sizes contain 16,000, 32,000 and 64,000 frames. We train 22 BPE vocabularies on the transcripts: 500, 1,000, 2,000 and 5,000 classes, followed by increments of 10,000 from 10,000 to 180,000.
Every count includes blank. For each batch size, augmented features are fixed across vocabularies, while retokenization can change target lengths and the number of selected classes.

\paragraph{Model and computation.} The Zipformer-M encoder produces 512-dimensional states in BF16 mixed precision. Both methods apply dropout and the vocabulary head in FP32, using PyTorch 2.9.1 and k2 1.24.4.dev20251118.
Standard CTC applies log-softmax followed by PyTorch's native CUDA CTC.
For Pruned CTC, we use the precision strategy of \Appref{app:impl:precision}, vocabulary chunk width 4096, alignment beam 100, and lattice limits of 100,000,000 states and 1,073,741,824 arcs. Both dynamic programs receive FP32 scores, and Pruned CTC computes the retained alignment log-likelihood in FP64.
The head-and-loss benchmark keeps its initial head weights fixed and uses identical FP32 encoder outputs and dropout masks for both methods, without propagating gradients into the encoder.
Full training steps start from the same initialization and process the same ordered batches. Each method updates its own weights with ScaledAdam and the Eden learning-rate schedule, using a base learning rate of 0.045.

\paragraph{Measurement.} Each configuration uses one NVIDIA-H100-SXM5-80GB GPU for three runs of five warmup and 50 measured batches.
We average run medians after computing memory and time ratios per paired batch.

\paragraph{Memory scaling.} On this grid, memory for the standard CTC head and loss is approximated by $(4MV + 640M)$ FP32 values, within 0.24\% of every measured point. This empirical model describes the benchmark's vocabulary-dependent allocations.
For Pruned CTC, the temporary $M \times C$ block used by the FP64 normalizer dominates memory once the vocabulary spans several chunks. The selected classes require at most 2,580 columns across the grid.
Standard CTC first runs out of memory in the head and loss at $V = 170{,}000$ for $M = 32{,}000$, and at $V = 90{,}000$ for $M = 64{,}000$.
The crosses in \Figref{fig:vocab-scaling} extrapolate a linear fit to the four largest feasible vocabularies in each column, giving 81.1 and 86.0 GiB at these two points.
For full training steps, the first failures occur at $V = 140{,}000$ at $M = 32{,}000$ and $V = 60{,}000$ at $M = 64{,}000$. \Tabref{tab:whole-step} reports the largest tested vocabulary at which standard CTC completes a training step for each batch size.

\subsection{LLM-based ASR}
\label{app:config:bench:llm}

\paragraph{Setup.} We profile Query-readout LLM-CTC with SPEAR-Qwen3-8B and SPEAR-Qwen3-14B on one NVIDIA A100-SXM4-80GB GPU, with activation checkpointing enabled for both.
The frozen LLM head has $V=151{,}936$, with $D=4096$ and $D=5120$, respectively.
\Appref{app:config:llm} and \Appref{app:config:spear} give the shared training and model settings. These measurements use a 1,200-second padded duration budget.
The GigaSpeech batch has $N=107$, $T=111$ and $M=11{,}458$, using the same augmented 128-dimensional FBank features and targets for every method.
The 8B and 14B LLM-CTC training runs in \Tabref{tab:spear-training-budgets} use eight GPUs and 200-second microbatches, with activation checkpointing enabled only for 14B.

\begin{table}[t]
\centering
\caption{Memory and runtime comparison for SPEAR-Qwen3. Head/loss reports memory above the resident head and inputs. Memory is in GiB and time is in seconds per full training step.}
\label{tab:llm-resources}
\begin{tabular}{@{}lrrr@{\hspace{18pt}}rrr@{}}
\toprule
& \multicolumn{3}{c}{8B} & \multicolumn{3}{c}{14B} \\
\cmidrule(lr){2-4}\cmidrule(lr){5-7}
Method & Head/loss & Full step & Time & Head/loss & Full step & Time \\
\midrule
Standard CTC & 29.21 & 56.06 & 10.23 & 29.79 & 72.71 & 15.73 \\
Unpruned reduced CTC & 1.21 & 34.67 & 11.04 & 1.33 & 52.82 & 16.66 \\
Pruned CTC ($\beta=100$) & 1.14 & 34.67 & 11.00 & 1.20 & 52.84 & 16.61 \\
\bottomrule
\end{tabular}
\end{table}

\paragraph{CTC computation.} Standard CTC applies the FP32 LLM head projection and log-softmax in $N\times T\times V$ layout, then transposes for PyTorch CTC.
Unpruned reduced CTC evaluates $\Ls$ with grouped logits $\mR$ and recomputation. Pruned CTC uses $\Ls_\beta$ with $\beta=100$.
Both reduced implementations use $C=4096$ and FP64 normalization. The head is stored in BF16 and converted to FP32 for projection, in full for standard CTC and in chunks for the reduced implementations.
Losses are normalized by the batch's target-token count. We use PyTorch 2.9.1 and k2 1.24.4.dev20260625.

\paragraph{Measurement.} Within each model, methods start from the same trained checkpoint with fresh AdamW states and update independently.
Each repeats the batch for five warmup steps, a capture step, three measured steps and a diagnostic trace step.
\Tabref{tab:llm-resources} reports the largest memory and median time over the three measured steps. Timing includes prepared FBank transfer through the optimizer update and gradient reset.
For the head and loss, methods share states $\tX$ and blank offset $\eta$ from the unpruned reduced run, with the same frame counts, frozen head and targets within each model. Two warmup passes precede three measured forward and backward passes.

\paragraph{Results.} \Tabref{tab:llm-resources} shows that Pruned CTC reduces full-step memory by 38.2\% for 8B and 27.3\% for 14B, with 7.5\% and 5.6\% more time than standard CTC, respectively.
Memory for the head and loss falls by a factor of 24.8 to 25.7, showing that the savings persist with a frozen native LLM head during LoRA adaptation. Model storage and activations elsewhere in the training step limit the full-step reduction.
Unpruned reduced CTC retains all valid alignments and already achieves nearly the same full-step memory and runtime as Pruned CTC. Exact vocabulary reduction and recomputation therefore account for most of the measured savings.
The implementation with finite-beam pruning further lowers memory for the head and loss, with little additional effect on full-step cost.

\section{Experimental Settings}
\label{app:config}

\subsection{Encoder-based ASR}
\label{app:config:nn}

The Pruned CTC systems in \Secref{sec:exp:nn} use 80-dimensional FBank features and train on NVIDIA A100-SXM4-80GB GPUs.
Training uses one GPU for 120 epochs on AISHELL-1 and two GPUs for 100 epochs on LibriSpeech and 60 epochs on GigaSpeech.

\subsection{Shared LLM-based ASR settings}
\label{app:config:llm}

\paragraph{Training.} We train on GigaSpeech using NVIDIA A100-SXM4-80GB GPUs, with MUSAN~\citep{snyder2015musan} and SpecAugment. The speech encoder and pretrained LLM parameters, including the LLM head, remain frozen. We train the projector and apply LoRA~\citep{hu2022lora} to attention query, key, value and output projections and feed-forward up, gate and down projections. CTC adds a scalar blank offset, and query readouts use one shared query whose entries are initialized from a zero-mean Gaussian distribution.

\begin{table}[t]
\centering
\caption{Shared LLM-based ASR settings. LR: learning rate.
Warmup is in optimizer updates.
For SPEAR-Qwen3-32B, the peak LR is 0.00005 for both LLM-CE and LLM-CTC.}
\label{tab:llm-shared-config}
\begin{tabular}[t]{lr}
\toprule
Setting & Value \\
\midrule
Optimizer & AdamW \\
Peak LR (default) & 0.0001 \\
Weight decay & 0.1 \\
Moment coefficients & 0.9, 0.95 \\
Epsilon & 0.000001 \\
Warmup & 200, linear \\
LR decay & Cosine \\
LR floor & 5\% of peak \\
Clip norm & 5 \\
Mixed precision & BF16 \\
\bottomrule
\end{tabular}
\hspace{2em}
\begin{tabular}[t]{lr}
\toprule
Setting & Value \\
\midrule
LoRA rank & 64 \\
LoRA $\alpha$ & 64 \\
LoRA dropout & 0.05 \\
Query std. & 0.02 \\
LLM head classes & 151,936 \\
Vocabulary chunk & 4,096 \\
Alignment beam & 100 \\
Lattice states & 100,000,000 \\
Qwen3 tokenizer size & 151,669 \\
Qwen3-ASR tokenizer size & 151,705 \\
\bottomrule
\end{tabular}
\end{table}

\paragraph{Targets.} Transcripts are lowercased. CE uses a user/assistant template and supervises the transcript followed by EOS. CTC uses the five-token prompt ``Transcribe the speech.'' Blank occupies the first unused row of the LLM head beyond the tokenizer vocabulary, and all head classes remain in the normalizer.
Training utterances have positive frame counts and fit the LLM context window.
CTC targets also satisfy $T_n\geq |y_n|+r_n$.

\paragraph{Blank initialization.} We initialize the blank offset from \Secref{sec:ctcllm:readouts} as $\eta=\log c_{\mathrm{bl}}$, where $c_{\mathrm{bl}}=|\{v\in\sV:\cos(\mW_{v,:},\mW_{v_{\mathrm{bl}},:})>0.999\}|$ includes the blank row itself. This heuristic raises the initial blank logit to account for probability mass shared with similarly directed head rows.

\subsection{SPEAR-Qwen3}
\label{app:config:spear}

\paragraph{Models.} SPEAR-XLarge v2~\citep{yang2026spear} takes 128-dimensional FBank features and produces 1280-dimensional embeddings at about 50 frames/s.
The projector applies layer normalization, stacks five frames, and uses two linear layers with hidden width 4096 and ReLU, producing about ten embeddings/s.
\Tabref{tab:qwen-config} lists the Qwen3~\citep{yang2025qwen3} backbones, and \Tabref{tab:spear-training-budgets} gives their training configurations.

\begin{table}[t]
\centering
\caption{Qwen3 backbone configurations for SPEAR-Qwen3.}
\label{tab:qwen-config}
\begin{tabular}{lrrc}
\toprule
Qwen3 & LLM layers & LLM width $D$ & Tied input/output weights \\
\midrule
0.6B & 28 & 1024 & Yes \\
1.7B & 28 & 2048 & Yes \\
4B & 36 & 2560 & Yes \\
8B & 36 & 4096 & No \\
14B & 40 & 5120 & No \\
32B & 64 & 5120 & No \\
\bottomrule
\end{tabular}
\end{table}

\paragraph{Training.} All runs use eight GPUs, 180,000 optimizer updates and ZeRO stage 1.
Audio/update is 1,600 seconds, the upper limit across GPUs and accumulated microbatches.

\begin{table}[t]
\centering
\caption{Training configurations for SPEAR-Qwen3. GA: gradient accumulation.
AC: activation checkpointing. Microbatch bounds audio duration in seconds per GPU.}
\label{tab:spear-training-budgets}
\begin{tabular}{lcccccc}
\toprule
& \multicolumn{3}{c}{\textbf{\textit{LLM-CE}}} & \multicolumn{3}{c}{\textbf{\textit{LLM-CTC}}} \\
\cmidrule(lr){2-4}\cmidrule(lr){5-7}
Qwen3 & Microbatch & GA & AC & Microbatch & GA & AC \\
\midrule
0.6B, 1.7B, 4B, 8B & 200 & 1 & No & 200 & 1 & No \\
14B & 100 & 2 & No & 200 & 1 & Yes \\
32B & 100 & 2 & Yes & 200 & 1 & Yes \\
\bottomrule
\end{tabular}
\end{table}

\paragraph{Decoding.} For WER evaluation, LLM-CE uses beam search with length penalty 1.0, and LLM-CTC uses prefix beam search with token penalty 0.3. Both use beam size 4 and no external LM.
AR generation stops at EOS or after 200 generated tokens.

\paragraph{Timing protocol.}
\label{app:config:rtf}
\Tabref{tab:rtf-settings} gives the timing settings. Timing uses the same trained checkpoints as the WER evaluation, with LoRA weights merged for inference. All systems use the same 1,024 duration-stratified GigaSpeech \textit{test} utterances, totaling 1.82 hours.
LLM-CE uses greedy generation with beam size 1, stopping at EOS or after 200 generated tokens.
LLM-CTC uses framewise greedy decoding followed by CTC collapse, with frame counts determined by the audio.
FBank extraction finishes before timing. Timing includes transfer of prepared CPU FBank features, the encoder, projector, LLM, LLM head and text decoding. Paired systems use the same GPU and utterance order, with GPU synchronization before and after each utterance.
For each pass, we divide total recognition time by total audio duration and report the median of the three ratios.

\begin{table}[t]
\centering
\caption{SPEAR-Qwen3 timing settings.}
\label{tab:rtf-settings}
\small
\begin{tabular}{@{}ll@{}}
\toprule
Setting & Value \\
\midrule
GPU model & NVIDIA-H100-SXM5-80GB \\
LLM weights & BF16 \\
Encoder/projector & BF16 mixed precision \\
LLM head & FP32 \\
Vocabulary chunk & 4,096 \\
Batch size & 1 utterance \\
CPU threads & 1 \\
Software & PyTorch 2.9.1, Transformers 4.57.6 \\
Warmup utterances & 32 \\
Timed passes & 3 \\
\bottomrule
\end{tabular}
\end{table}

\subsection{Qwen3-ASR}
\label{app:config:qwen-asr}

\paragraph{Models and inputs.} \Tabref{tab:qwen-asr-config} supplements the shared settings in \Tabref{tab:llm-shared-config}. The projector has two linear layers separated by GELU, and the frozen speech encoder remains in evaluation mode.
Offline and streaming modes use 128-dimensional log-mel FBank from 16 kHz audio, with a 25 ms window and 10 ms hop.
Training targets are nonempty token sequences without special tokens.

\begin{table}[t]
\centering
\caption{Qwen3-ASR configurations for offline and streaming fine-tuning.
GA: gradient accumulation. AC: activation checkpointing.
Microbatch bounds audio duration in seconds per GPU. Audio/update is the upper limit in seconds across GPUs and accumulated microbatches.}
\label{tab:qwen-asr-config}
\begin{tabular}[t]{lcc}
\toprule
Setting & 0.6B & 1.7B \\
\midrule
Encoder blocks & 18 & 24 \\
Encoder width & 896 & 1024 \\
LLM layers & 28 & 28 \\
LLM width $D$ & 1024 & 2048 \\
\bottomrule
\end{tabular}
\hspace{2em}
\begin{tabular}[t]{lcc}
\toprule
Setting & 0.6B & 1.7B \\
\midrule
GPUs & 8 & 8 \\
Microbatch & 200 & 200 \\
GA & 1 & 1 \\
Audio/update & 1,600 & 1,600 \\
Updates & 180,000 & 180,000 \\
AC & No & No \\
ZeRO stage & 1 & 1 \\
\bottomrule
\end{tabular}
\end{table}

\paragraph{Offline inputs.} Feature extraction, normalization and SpecAugment operate on entire utterances. The encoder retains its pretrained 1-second convolution chunks and 8-second bidirectional attention windows. All speech embeddings precede queries under causal attention in the LLM.

\paragraph{Streaming inputs.} Feature extraction, normalization and SpecAugment operate independently on each 2-second chunk. The encoder processes the preceding $h$ seconds together with the current chunk using bidirectional attention and returns only current-chunk outputs. Training batches the overlapping encoder windows in \Figref{fig:streaming-readout}\captionb{} as independent sequences.
Each full chunk yields 200 mel frames, 26 speech embeddings and 26 queries. Final partial chunks retain their actual lengths, with padding only as required by the analysis window and hop.
Each training microbatch shares one integer $h$ across utterances and devices.
We select 8 seconds with probability one half and each value from 1 to 7 seconds with probability $1/14$. The encoder and LLM attention mask use the same $h$.
Validation uses 8 seconds. Decoding reuses each streaming checkpoint across the histories in \Tabref{tab:streaming}, with at most 10 seconds of encoder input per chunk.
History boundaries follow audio time and may split an earlier chunk, as defined in \Appref{app:readout:streaming}.

\paragraph{Decoding.} Both offline and streaming modes use token penalty 0.7.
Streaming prefix beam search retains candidate prefixes and their blank/nonblank probabilities, so CTC collapse $\gB$ spans chunk boundaries.
Greedy collapse retains the preceding raw label, including blank.
Each utterance maintains independent mel history, KV cache and CTC state.
Accumulated token IDs are decoded together, and states are reset when the utterance ends.
Evaluation retains utterances independently of reference token counts.

\subsection{External ASR baselines}
\label{app:config:baselines}

\paragraph{Scoring.} All models use 16 kHz mono audio and beam size 4.
Model-specific prompts and language IDs are given below.
References use official GigaSpeech normalization.
For hypotheses, we convert standalone numeric literals to English cardinal words with num2words 0.5.14. We remove non-scoring tags and punctuation while preserving word-internal apostrophes and splitting hyphens, then apply official GigaSpeech normalization.

\begin{itemize}[leftmargin=*, itemsep=0pt, topsep=2pt]
\item \textbf{Whisper:} The large-v3 model\footnote{\urlstyle{same}\url{https://huggingface.co/openai/whisper-large-v3}} is decoded with the language ID fixed to English.
\item \textbf{Qwen2-Audio:} The base Qwen2-Audio-7B\footnote{\urlstyle{same}\url{https://huggingface.co/Qwen/Qwen2-Audio-7B}} uses the official English prompt, ``Detect the language and recognize the speech:'', followed by the English language token.
\item \textbf{Qwen3-ASR:} The 0.6B\footnote{\urlstyle{same}\url{https://huggingface.co/Qwen/Qwen3-ASR-0.6B}} and 1.7B\footnote{\urlstyle{same}\url{https://huggingface.co/Qwen/Qwen3-ASR-1.7B}} models are decoded using the official qwen-asr Python package, with the language ID fixed to English.
\end{itemize}

\section{LLM-CTC Analysis}
\label{app:ablation}
\label{app:readout}

\subsection{Context and computation}
\label{app:readout:theory}
\label{app:readout:cost}
\label{app:readout:streaming}

\paragraph{Readout structure and offline cost.} \Tabref{tab:readout-structure} compares the speech and query readouts defined in \Secref{sec:ctcllm}.
For a fixed number $T_n$ of speech embeddings, all provide $T_n$ CTC frames and the same valid alignment set for a given transcript.
Speech embeddings already carry encoder context.
Query copies share their initial vector but occupy distinct rotary positions: the distance from query $j$ to speech position $t$ is $T_n+j-t$ offline and $T_{n,b}+j-t$ within streaming chunk $b$.
Each LLM layer performs tokenwise work proportional to its input length $S$, with $S(S+1)/2$ attention pairs under causal attention over the full sequence.
The offline query readout increases $S$ from $P+T_n$ to $P+2T_n$.
Both offline readouts incur the same encoder and projector work.
All readouts apply the frozen LLM head to $T_n$ normalized CTC states.

\begin{table}[t]
\centering
\caption{CTC readout structure for $T_n$ speech embeddings, prompt length $P$ and LLM width $D$. Speech, Query and Streaming query denote the speech, query and streaming query readouts.
Position counts are per utterance. Forward passes refer to inference. Speech access describes direct attention within the LLM, with streaming left history bounded by $h$ seconds.}
\label{tab:readout-structure}
\begin{tabular}{lccc}
\toprule
Property & Speech & Query & Streaming query \\
\midrule
LLM positions & $P+T_n$ & $P+2T_n$ & $P+2T_n$ \\
CTC frames & $T_n$ & $T_n$ & $T_n$ \\
Query parameters & 0 & $D$ & $D$ \\
Speech access & Prefix & Entire utterance & Current chunk + history \\
Initial content & Speech embedding & Shared query & Shared query \\
LLM forward passes & One per utterance & One per utterance & One per chunk \\
\bottomrule
\end{tabular}
\end{table}

\paragraph{Streaming time coordinates.} For the Qwen3-ASR frontend in \Appref{app:config:qwen-asr}, let $s_{n,b}$ be the waveform sample count of nonempty chunk $b$ and $F_{n,b}$ its mel frame count after required end padding.
With 16 kHz audio, a 400-sample analysis window and a 160-sample hop,
\begin{equation*}
  F_{n,b} = \left\lceil\frac{\max(s_{n,b},400)}{160}\right\rceil,
  \qquad
  T_{n,b} = 13\left\lfloor\frac{F_{n,b}}{100}\right\rfloor
  + \left\lceil\frac{F_{n,b}\bmod 100}{8}\right\rceil .
\end{equation*}
Chunk $b$ starts at $a_{n,b}=2(b-1)$ seconds from the utterance start under the 2-second schedule. For $j=1,\dots,T_{n,b}$, its speech position and corresponding query position in \Eqref{eq:streaming-readout} share the nominal time coordinate
\begin{equation*}
  \tau_n(o_{n,b}+j) = \tau_n(o_{n,b}+T_{n,b}+j)
  = a_{n,b} + \left\lfloor\frac{j-1}{13}\right\rfloor
    + 0.08\bigl((j-1)\bmod 13\bigr) .
\end{equation*}
These coordinates, in seconds, follow the frontend geometry and govern cache selection without imposing token alignments or emission times.
A final 10 ms chunk yields three mel frames after padding and one CTC frame.

\paragraph{Streaming attention.} For the mask in \Figref{fig:streaming-readout}\captiona{}, let $\sJ_n=\{P+1,\dots,P+2T_n\}$ contain the valid non-prompt sequence positions, and let $b_n(k)$ give the chunk containing position $k$.
For position $i\in\sJ_n$ in chunk $b$, the visible key positions at every LLM layer are
\begin{equation*}
  \begin{aligned}
    \gV_n(i) ={}& \{1,\dots,P\} \\
    &{}\cup\{k\in\sJ_n:b_n(k)<b,\ a_{n,b}-h\leq\tau_n(k)<a_{n,b}\} \\
    &{}\cup\{k\in\sJ_n:b_n(k)=b,\ k\leq i\} .
  \end{aligned}
\end{equation*}
A prompt position $i$ sees only prompt positions $1,\dots,i$, and batch padding is excluded. The history cutoff is fixed at the chunk start for all its positions.
For example, a 1-second history retains the last 13 speech positions and the last 13 query positions of the preceding full chunk, which occupy two separate ranges in the concatenated sequence.
After chunk $b$, the KV cache keeps the prompt and non-prompt entries with $\tau_n(k)\geq a_{n,b}+0.01F_{n,b}-h$.
For a full chunk, this is the next chunk's history cutoff.
Retained entries keep their original sequence indices and rotary positions.

\paragraph{Computation equivalence.} Consider parallel batch computation and per-utterance cached computation with fixed speech embeddings, parameters, rotary frequencies and $h$, and with dropout disabled.
The prompt states agree because prompt attention is causal. Suppose the cached states from earlier chunks match those from parallel batch computation at every LLM layer. The first layer of the current chunk receives the same speech and query embeddings. If a layer's current input states agree, its query, key and value projections agree. Together with the matching historical keys and values, these give the same attention output because both computations use $\gV_n(i)$ and the original rotary positions.
The positionwise transformations then give identical next-layer states.
Induction over LLM layers and chunks proves equality of query states in exact arithmetic.
Cache eviction preserves this argument because it removes only entries outside the next chunk's visible history.

\paragraph{Streaming computation.} Let $A$ bound the number of speech embeddings per audio chunk and $J$ the number of historical speech and query positions retained before a chunk. The LLM processes at most $2A$ new positions per chunk.
Excluding the one-time prompt computation, each LLM layer has at most $2A(P+J)+A(2A+1)$ attention pairs.
The speech encoder recomputes its bounded left history with the current chunk, while the LLM reuses its KV cache.
Each new query still requires full-vocabulary projection.
For fixed chunk and history lengths, attention work per chunk and cache size remain bounded as the utterance grows. Retained states can carry earlier context through the LLM layers.
For the 2-second chunks and 8-second history in \Appref{app:config:qwen-asr}, $A=26$ and $J=208$. The KV cache holds at most $P+208$ positions between chunks and $P+260$ while processing a chunk. These bounds concern attention storage, separately from model weights, CTC search state and accumulated text.

\subsection{Streaming consistency checks}
\label{app:readout:consistency}

\paragraph{Setup.} For Qwen3-ASR 0.6B and 1.7B, we compare parallel batch computation of all chunks with per-utterance cached computation that processes chunks sequentially using a KV cache, as in \Figref{fig:streaming-readout}\captionb{} and \captionc{}.
Both follow the same streaming attention rule.
All conditions use an FP32 LLM head on NVIDIA A100-SXM4-80GB GPUs, with dropout and augmentation disabled. The same eight GigaSpeech training batches use the training sampler's padded 200-second budget. Chunks are 2 seconds, and integer left histories range from 1 to 8 seconds.

\begin{figure}[t]
\centering
\includegraphics[width=\linewidth]{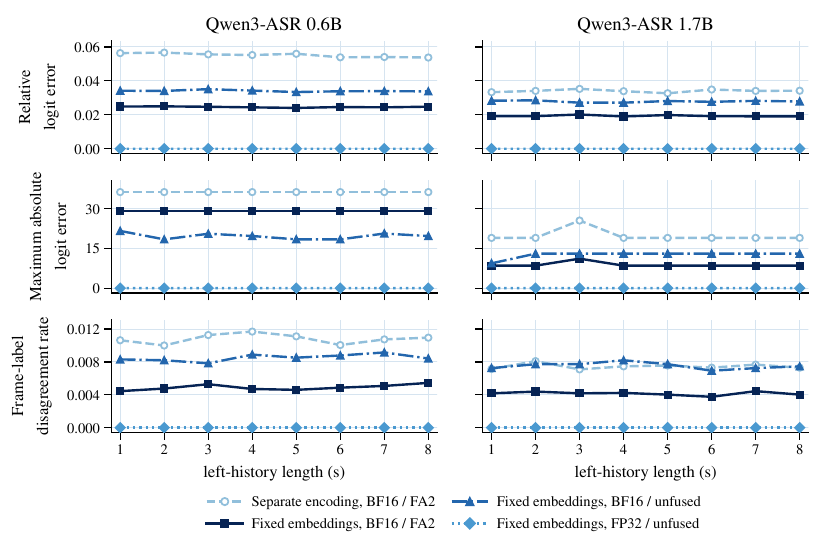}
\caption{Comparing parallel batch computation with per-utterance cached computation on GigaSpeech.
Rows show relative and maximum absolute logit errors and the greedy frame-label disagreement rate.
FA2 stands for FlashAttention~2.}
\label{fig:streaming-consistency}
\end{figure}

\paragraph{Input and precision controls.} The baseline uses BF16 mixed precision and FlashAttention 2.8.3, with speech embeddings computed separately for the two computation modes.
Fixed-input comparisons reuse speech embeddings saved once per history in both LLM computations, preserving parameter values, original positions, rotary frequencies and attention visibility.
We evaluate fixed inputs with BF16 FlashAttention and with an unfused attention implementation in BF16 and FP32, which computes matrix products and softmax separately.

\paragraph{Results.} For each model, condition and history, we pool valid query positions across the eight batches. \Figref{fig:streaming-consistency} reports relative full-vocabulary logit error, $\|\mZ_{\mathrm{cache}}-\mZ_{\mathrm{par}}\|_F/\|\mZ_{\mathrm{par}}\|_F$, maximum entrywise absolute error over all head classes, and the greedy frame-label disagreement rate.
The baseline's relative logit error ranges from 0.0538 to 0.0567 for 0.6B and from 0.0327 to 0.0353 for 1.7B.
Fixed embeddings yield ranges of 0.0241 to 0.0250 and 0.0190 to 0.0202, respectively.
With fixed inputs and BF16 unfused attention, the ranges are 0.0334 to 0.0351 and 0.0271 to 0.0285.
FP32 unfused attention yields relative errors of 0.00000299 to 0.00000489 across both models, with a maximum absolute error of $6.41\times10^{-3}$ and no greedy frame-label disagreements.
Fixing the embeddings improves BF16 agreement, and the remaining differences depend strongly on LLM precision and attention implementation.

\end{document}